\documentclass{article}
\usepackage{common}

\title{Direct Sampling with a Step Function}
\author{Andrew~M. Raim
\vspace{0.5em} \\
Center for Statistical Research and Methodology, U.S. Census Bureau
}
\date{}

\begin{document}

\maketitle

\begin{abstract}
The direct sampling method proposed by \citet[JCGS][]{WalkerEtAl2011} can generate draws from weighted distributions possibly having intractable normalizing constants. The method may be of interest as a useful tool in situations which require drawing from an unfamiliar distribution. However, the original algorithm can have difficulty producing draws in some situations. The present work restricts attention to a univariate setting where the weight function and base distribution of the weighted target density meet certain criteria. Here, a variant of the direct sampler is proposed which uses a step function to approximate the density of a particular augmented random variable on which the method is based. Knots for the step function can be placed strategically to ensure the approximation is close to the underlying density. Variates may then be generated reliably while largely avoiding the need for manual tuning or rejections. A rejection sampler based on the step function allows exact draws to be generated from the target with lower rejection probability in exchange for increased computation. Several applications of the proposed sampler illustrate the method: generating draws from the Conway-Maxwell Poisson distribution, a Gibbs sampler which draws the dependence parameter in a random effects model with conditional autoregression structure, and a Gibbs sampler which draws the degrees-of-freedom parameter in a regression with t-distributed errors.
\end{abstract}


\blfootnote{%
Disclaimer: This article is released to inform interested parties of ongoing research and to encourage discussion of work in progress. Any views expressed are those of the author and not those of the U.S.~Census Bureau.
\begin{flushleft}
For correspondence: \\
Andrew~M. Raim (\url{andrew.raim@census.gov}) \\
Center for Statistical Research and Methodology \\
U.S.~Census Bureau \\
Washington, DC, 20233, U.S.A.
\end{flushleft}
}

\section{Introduction}
\label{sec:intro}

This paper revisits the direct sampling method proposed by \citet{WalkerEtAl2011}. Consider drawing a random variable $X$ with support $\Omega \subseteq \mathbb{R}$ whose density takes the form
\begin{align}
f(x) = w(x) g(x) / \psi, \quad x \in \Omega,
\quad \psi = \int_\Omega w(x) g(x) d\nu(x),
\label{eqn:weighted-density}
\end{align}
where $\nu(\cdot)$ is a dominating measure. The distribution of $x$ may be discrete, continuous, or continuous with point masses. Density $f$ can be recognized as a weighted distribution \citep[e.g.][]{PatilRao1978} with a weight function $w : \mathbb{R} \rightarrow [0, \infty)$ which adjusts the base density $g$ in some prescribed way. Direct sampling augments a random variable $U$ so that the joint distribution of $[X, U]$ is easier to draw than $X$ itself. Let $\ind(\cdot)$ be the indicator function and suppose $c = \sup_{x \in \Omega} w(x)$ is finite. Assume that $[U \mid X = x] \sim \text{Uniform}(0, w(x)/c)$, so that
\begin{align*}
f(u \mid x) = \frac{c}{w(x)} \ind(0 < u < w(x)/c).
\end{align*}
Define the event $A_u = \{ x \in \Omega : w(x) > u c \}$. The joint density of $[X,U]$ is then
\begin{align}
f(x,u) = \frac{c}{\psi} g(x) \ind(x \in A_u).
\label{eqn:joint-augmented}
\end{align}
%
%
%
From \eqref{eqn:joint-augmented}, the marginal density of $U$ may be obtained as
\begin{align*}
p(u) = \frac{c}{\psi} \Prob(A_u),
\quad u \in [0, 1],
\end{align*}
with $\Prob(A_u) = \int \ind(x \in A_u) g(x) d\nu(x)$. The distribution of $[X \mid U = u]$ is then
\begin{align}
f(x \mid u) = \frac{g(x)}{\Prob(A_u)} \ind(x \in A_u).
\label{eqn:fx-truncated}
\end{align}
Now $U$ is bounded in $[0,1]$, with $A_s \supseteq A_u$ if $s \leq u$ so that $\Prob(A_u)$ is monotonically nonincreasing in $u$. Evaluated at the endpoints $\{0,1\}$, $A_0$ is equivalent to the support of $w$ with $\Prob(A_0) = \int_{\Omega} \ind(w(x) > 0) g(x) d\nu(x)$ and $A_1$ is an empty set with $\Prob(A_1) = 0$.

A draw from $f(x)$ may be approximately obtained by drawing $U$ from $p(u)$ then $X$ from $f(x \mid u)$ in the following way. For a predefined positive integer $N$, compute
\begin{align}
q(k/N) = \frac{\Prob(A_{k/N})}{\sum_{\ell=0}^N \Prob(A_{\ell/N})}, \quad k = 0, 1, \ldots, N.
\label{eqn:knots}
\end{align}
Sample discrete random variable $K$ from the values $0, 1, \ldots, N$ with respective probabilities $q(0/N), \ldots, q(N/N)$, then draw from $[U \mid K = k] \sim \text{Beta}(k+1, N-k+1)$. The marginal density of $U$ is then proportional to
\begin{align}
&\sum_{k=0}^N \frac{u^k (1-u)^{N-k}}{B(k+1, N-k+1)} q(k/N) \nonumber \\
&\propto \sum_{k=0}^N \binom{N}{k} u^k (1-u)^{N-k} q(k/N),
\quad  u \in (0,1)
\label{eqn:bernstein}
\end{align}
where $B(a, b) = \Gamma(a) \Gamma(b) / \Gamma(a + b)$ is the beta function. Expression \eqref{eqn:bernstein} is an approximation to $p(u)$ by Bernstein polynomials \citep[e.g.][]{Rivlin1981}. A variate $x$ from the truncated distribution \eqref{eqn:fx-truncated} may be obtained by repeating draws of candidate $x^*$ from $g(x)$, which is straightforward in many applications, until $x^* \in A_u$ where $x$ is taken to be $x^*$. This algorithm was described by \citet{WalkerEtAl2011} as a basic implementation of their direct sampling idea.

Direct sampling is interesting as an alternative to standard strategies such as Metropolis-Hastings, slice sampling, rejection sampling, and adaptive rejection sampling \citep[e.g.][]{RobertCasella2004}; however, it does not yet appear to be widely adopted the in literature. An exception is \citet{BraunDamien2016}, who explore it as a scalable replacement for the inherently serial Markov chain Monte Carlo (MCMC) approach to Bayesian computing. The sampler described thus far may encounter challenges in practice which prevent it from successfully drawing from the target distribution. First, the basic rejection sampling method described to draw from \eqref{eqn:fx-truncated} may require a very large number of candidates when the set $A_u$ has small probability under $g(x)$. Second, the function $\Prob(A_u)$ may take large sudden steps not efficiently captured by polynomial approximation. Third, the support of $\Prob(A_u)$ may be concentrated on an interval $[u_L, u_H]$ with $u_H$ a very small positive number. For example, the second and third issues are seen in the bottom row of Figure~\ref{fig:cmp_step}.

It is possible to focus the Bernstein approximation~\eqref{eqn:bernstein} to the interval $[u_L,u_H]$ or consider other functional bases from the literature. However, we take an approach based on step functions which are relatively simple with low computational burden. \citet[Section~3.6]{MartinoEtAl2018} provide background on step functions in the context of rejection sampling; this motivates our use in approximating and drawing from the distribution $p(u)$. Expressions for the density, cumulative distribution function (CDF), and quantile function are available, and exact draws may be taken directly via the quantile function. Through appropriate placement of knot points, a step function can directly capture any jumps encountered in $p(u)$. Simple bounds on the accuracy of the approximation can be obtained in our setting, and such bounds can be improved by placing additional knot points until a desired tolerance is achieved. A step function can serve as an envelope in rejection sampling if exact draws from $f(x)$ are required. In addition to assuming univariate $X$, we restrict ourselves to weight functions $w$ where $A_u$ is an interval for each $u \in [0,1]$. Ideally, endpoints of $A_u$ and the CDF and quantile function of base distribution $g$ are readily computed.

The remainder of the paper proceeds as follows. Section~\ref{sec:truncated-base} discusses generating draws from \eqref{eqn:fx-truncated} in this setting without rejections. Section~\ref{sec:step-function} presents use of the step function in direct sampling. Section~\ref{sec:examples} considers three illustrative applications using this formulation of the direct sampler: drawing from the Conway-Maxwell Poisson distribution, a Gibbs sampler for a conditional autoregression random effects model including inference on the dependence parameter, and a Gibbs sampler for a regression model with errors following a Student's t-distribution including inference on the degrees of freedom. Finally, Section~\ref{sec:discussion} concludes the paper. Supporting code is provided as an electronic supplement, including materials to replicate the examples, implemented in both pure R \citep{Rcore2022} and with integrated C++ via the Rcpp framework \citep{Eddelbuettel2013}.

\section{Drawing from the Truncated Base Distribution}
\label{sec:truncated-base}
We first consider efficiently drawing from \eqref{eqn:fx-truncated}. Suppose density $g$ is associated with CDF and quantile functions
\begin{align*}
G(x) = \int_{-\infty}^x g(s) d\nu(s),
\quad
G^{-}(\varphi) = \inf\{ x \in \Omega : G(x) \geq \varphi \},
\end{align*}
respectively. With $A_u$ assumed to be an interval $(x_1(u), x_2(u))$, whose endpoints are identified by the roots of the equation $w(x) = cu$, \eqref{eqn:fx-truncated} represents the base distribution $g$ truncated to the interval $(x_1(u), x_2(u))$ with
\begin{align*}
f(x \mid u) = \frac{g(x)}{t-s} \cdot \ind(x_1(u) < x < x_2(u)),
\end{align*}
where $G(x-) = \lim_{t \uparrow x} G(t)$, $s = G(x_1(u))$ and $t = G(x_2(u)-)$ are CDF values evaluated at the endpoints, and $\lceil x \rceil$ and $\lfloor x \rfloor$ represent the ceiling and floor functions of $x$, respectively. The associated CDF of $[X \mid U = u]$ is
\begin{align}
F(x \mid u) =
\frac{G(x) - s}{t-s},
\quad x_1(u) < x < x_2(u),
\label{eqn:truncated-cdf}
\end{align}
with $F(x \mid u) = 0$ for $x < x_1(u)$ and $F(x \mid u) = 1$ for $x > x_2(u)$. We may invert $F(x \mid u)$ to obtain the $\varphi \in (0,1)$ quantile of $[X \mid U = u]$ as
\begin{align}
F^{-}(\varphi \mid u)
&= \inf\{ x \in A_u : F(x \mid u) \geq \varphi \} \nonumber \\
&= \inf\{ x \in \Omega : F(x \mid u) \geq \varphi \} \label{eqn:supp-expansion} \\
&= \inf\{ x \in \Omega : [G(x) - s] / (t-s) \geq \varphi \} \nonumber \\
&= \inf\{ x \in \Omega : G(x) \geq (t-s) \varphi + s \} \nonumber \\
&= G^{-}((t-s) \varphi + s).
\label{eqn:truncated-quantile}
\end{align}
To justify step \eqref{eqn:supp-expansion}, $x \in \Omega \setminus A_u$ implies $x \leq x_1(u)$ or $x \geq x_2(u)$ so that either $F(x \mid u) = 0$ and does not satisfy the criteria $F(x \mid u) \geq \varphi$, or $F(x \mid u) = 1$ and there is a smaller $x' \in (x_1(u), x_2(u))$ with $F(x' \mid u) \geq \varphi$. Therefore, including $\Omega \setminus A_u$ does not change the infimum. Now, an exact draw can be obtained via the inverse CDF method \citep[e.g.][Section~22.3]{Lange2010} using $x = F^{-}(V \mid u)$ with $V \sim \text{Uniform}(0,1)$.

\section{Step Function}
\label{sec:step-function}

To approximate the density $p(u)$, we first identify an interval $[u_L,u_H] \subseteq [0,1]$ which contains the ``descent'' from its maximum value to a value of zero; any further effort should be focused within this interval. Define $u_L$ as the smallest number such that $\Prob(A_{u_L}) < \Prob(A_0)$ for the unnormalized density and $u_H$ as the smallest number such that $\Prob(A_{u_H}) > 0$. A bisection method described in may be used to locate $u_L$ and $u_H$; see Remark~\ref{remark:bisection} at the end of this section.

To approximate the unnormalized $\Prob(A_u)$, let $u_0 < \cdots < u_N$ be knot points with $u_0 = u_L$ and $u_N = u_H$ and consider the function
\begin{align*}
h^*(u) = \Prob(A_{u_0}) \cdot \ind(0 \leq u < u_0) + \sum_{j=0}^{N-1} \Prob(A_{u_j}) \cdot \ind(u_j \leq u < u_{j+1}).
\end{align*}
A density is obtained using $h(u) = h^*(u) / a$ with
\begin{align*}
a = \int_0^1 h^*(u) du
= \Prob(A_{u_0}) \cdot u_0 + \sum_{j=0}^{N-1} \Prob(A_{u_j}) \cdot (u_{j+1} - u_j).
\end{align*}
The corresponding CDF is the piecewise linear function
\begin{align*}
H(u) = a^{-1} \Prob(A_{u_0}) u, \quad \text{if $0 \leq u < u_0$},
\end{align*}
and
\begin{align*}
H(u) = &a^{-1} \Prob(A_{u_0}) u_0 + a^{-1} \sum_{j=0}^{\ell-1} \Prob(A_{u_j}) \cdot (u_{j+1} - u_j) \\
&+ a^{-1} \Prob(A_{u_\ell}) \cdot (u - u_\ell),
\end{align*}
if $u_\ell \leq u < u_{\ell+1}$ for $\ell \in \{ 0, \ldots, N-1 \}$, $H(u) = 0$ if $u \leq 0$ and $H(u) = 1$ if $u \geq u_N$. The quantile function is also a piecewise linear function,
\begin{align*}
H^{-1}(\varphi) &= u_\ell + (u_{\ell+1} - u_{\ell})
\frac{
\varphi - H(u_\ell)
}{
H(u_{\ell+1}) - H(u_\ell)
}
\end{align*}
for $H(u_\ell) \leq \varphi < H(u_{\ell+1})$, $\ell \in \{ 0, \ldots, N-1 \}$. A draw from $h$ can now be generated by $u = H^{-1}(V)$ where $V \sim \text{Uniform}(0,1)$.

The following proposition shows that the closeness of $h(u)$ to $p(u)$ can be characterized by total variation distance. Let $\mathcal{R}_{j}$ represent the rectangle in $\mathbb{R}^2$ whose upper-left point is $(u_{j-1}, \Prob(A_{u_{j-1}}))$ and lower-right point is $(u_j, \Prob(A_{u_j}))$, for $j = 1, \ldots, N$. The area of $\mathcal{R}_{j}$ is $|\mathcal{R}_{j}| = \left[ \Prob(A_{u_{j-1}}) - \Prob(A_{u_j}) \right] (u_j - u_{j-1})$.

\begin{proposition}
\label{result:rectangles}
Let $\mathcal{B}$ denote the collection of measurable subsets of $[0,1]$; then
\begin{align}
\sup_{B \in \mathcal{B}}\left\lvert \int_B h(u) du - \int_B p(u) du \right\rvert \leq \frac{c}{\psi} \sum_{j=1}^N |\mathcal{R}_j|.
\label{eqn:rectangles}
\end{align}

\end{proposition}

\begin{proof}
First considering the unnormalized densities,
\begin{align}
&\sup_{B \in \mathcal{B}}\left\lvert \int_B h^*(u) du - \int_B \Prob(A_u) du \right\rvert \\
&\quad= \sup_{B \in \mathcal{B}} \int_B [h^*(u) - \Prob(A_u)] du \nonumber \\
&\quad= \int_0^1 [h^*(u) - \Prob(A_u)] du \nonumber \\
&\quad= \sum_{j=1}^N \int_{u_{j-1}}^{u_j} [h^*(u) - \Prob(A_u)] du \nonumber \\
&\quad\leq \sum_{j=1}^N \int_{u_{j-1}}^{u_j} [\Prob(A_{u_{j-1}}) - \Prob(A_{u_j})] du \nonumber \\
&\quad= \sum_{j=1}^N [\Prob(A_{u_{j-1}}) - \Prob(A_{u_j})] (u_j - u_{j-1})
= \sum_{j=1}^N |\mathcal{R}_j|.
\label{eqn:unnorm-diff}
\end{align}
We have used the fact that $h^*(u) = \Prob(A_u)$ for $u \in [0, u_0]$ and $h^*(u) \geq \Prob(A_u)$ otherwise. Integrating each term of the inequality $\Prob(A_u) \leq h^*(u) \leq 1$ over $u \in [0,1]$ gives
\begin{align}
\psi/c \leq a \leq 1
\quad \iff \quad
1 \leq 1/a \leq c/\psi.
\label{eqn:norm-const-ineq}
\end{align}
Combining this with \eqref{eqn:unnorm-diff} yields inequalities for the normalized densities
\begin{align}
\int_B h(u) du - \int_B p(u) du
&\leq \frac{c}{\psi} \left[ \int_B h^*(u) du - \int_B \Prob(A_u) du \right] \nonumber \\
&\leq \frac{c}{\psi} \sum_{j=1}^N |\mathcal{R}_j|.
\label{eqn:norm-diff1}
\end{align}
and
\begin{align}
\int_B p(u) du - \int_B h(u) du
&\leq \frac{c}{\psi} \int_B h^*(u) du - \frac{1}{a} \int_B h^*(u) du \nonumber \\
&= \left[ \frac{a - \psi/c}{a \psi/c} \right] \int_B h^*(u) du \nonumber \\
&\leq \left[ \frac{a - \psi/c}{a \psi/c} \right] a \nonumber \\
&= \frac{c}{\psi} \left[ \int_0^1 h^*(u) du - \int_0^1 \Prob(A_u) du \right] \nonumber \\
&\leq \frac{c}{\psi} \sum_{j=1}^N |\mathcal{R}_j|.
\label{eqn:norm-diff2}
\end{align}
The result follows from \eqref{eqn:norm-diff1} and \eqref{eqn:norm-diff2}.
\end{proof}

The upper bound in \eqref{eqn:rectangles} is seen as a product of two factors: a constant $c/\psi$ which is determined by the density \eqref{eqn:weighted-density} of interest and $\sum_{j=1}^N |\mathcal{R}_j|$ which can be influenced by the selection of knots.

There are a number of possible choices for the knots $u_1, \ldots, u_{N-1}$. Equally-spaced knots $u_j = u_L + (j/N)(u_H - u_L)$ provide simplicity but can fail to capture regions of $[u_L,u_H]$ with sudden changes in $\Prob(A_u)$. Proposition~\ref{result:rectangles} motivates placement of knots to ensure that no $\mathcal{R}_j$ is too large. Namely, given $u_0, u_1, \ldots, u_k$ with associated rectangles $\mathcal{R}_1, \ldots, \mathcal{R}_k$ we consider placing a new knot $u_*$ at the midpoint of $[u_{j-1}, u_j]$ which has the largest $|\mathcal{R}_j|$. This replaces $\mathcal{R}_j$ with new rectangles $\mathcal{R}_j^{(1)}$ and $\mathcal{R}_j^{(2)}$, yielding an improvement $|\mathcal{R}_j^{(1)}| + |\mathcal{R}_j^{(2)}| < |\mathcal{R}_j|$ in regions where $\Prob(A_u)$ is decreasing; otherwise, $|\mathcal{R}_j^{(1)}| + |\mathcal{R}_j^{(2)}| = |\mathcal{R}_j|$ so that the bound in \eqref{eqn:rectangles} is no worse. Stated as Algorithm~\ref{alg:small_rects}, this method often provides a better selection of knots under a fixed $N$ than equally-spaced points, at the cost of increased computation. Use of a data structure such as a priority queue \citep[Section~6.5]{CormenEtAl2009} can help to avoid repeated sorting of $|\mathcal{R}_1|, \ldots, |\mathcal{R}_k|$. An illustration of one step of Algorithm~\ref{alg:small_rects} is shown in Figure~\ref{fig:approx-error}.

\begin{remark}[Midpoint]
\label{remark:midpoint}
The midpoint in Algorithm~\ref{alg:small_rects} is specified by a function $\text{mid}(x,y)$, with typical choices being the arithmetic mean $\text{mid}(x,y) = (x+y)/2$ or the geometric mean $\text{mid}(x,y) = (xy)^{1/2}$. The arithmetic mean may yield a better approximation when $\Prob(A_u) \gg 0$ on a large potion of $[0,1]$. However, the geometric mean may be preferred when some knots are extremely small. For example, if $u_L = 10^{-100}$ and $u_H = 10^{-10}$ and a large descent occurs near $u_L$, the geometric mean $10^{-55}$ is much closer to the descent than the arithmetic mean $\frac{1}{2} 10^{-100} + \frac{1}{2} 10^{-10} \approx \frac{1}{2} 10^{-10}$. An example where knots are needed very close to zero is given in Section~\ref{sec:example-cmp}. The geometric mean is assumed for the remainder of the paper unless otherwise noted.
\end{remark}

\begin{figure*}
\centering
\begin{subfigure}{0.40\textwidth}
\includegraphics[width=\textwidth, trim=1cm 0cm 0cm 0cm, clip]{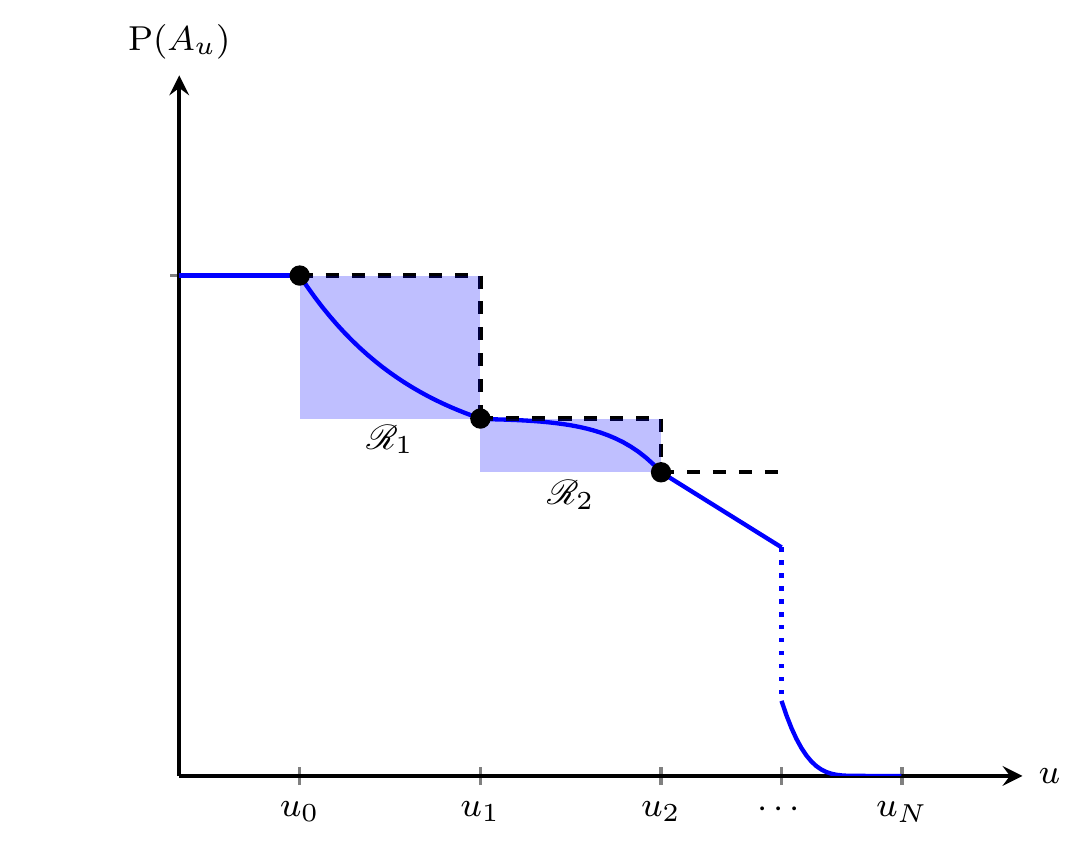}
\caption{}
\label{fig:approx-error1}
\end{subfigure}
\begin{subfigure}{0.40\textwidth}
\includegraphics[width=\textwidth, trim=1cm 0cm 0cm 0cm, clip]{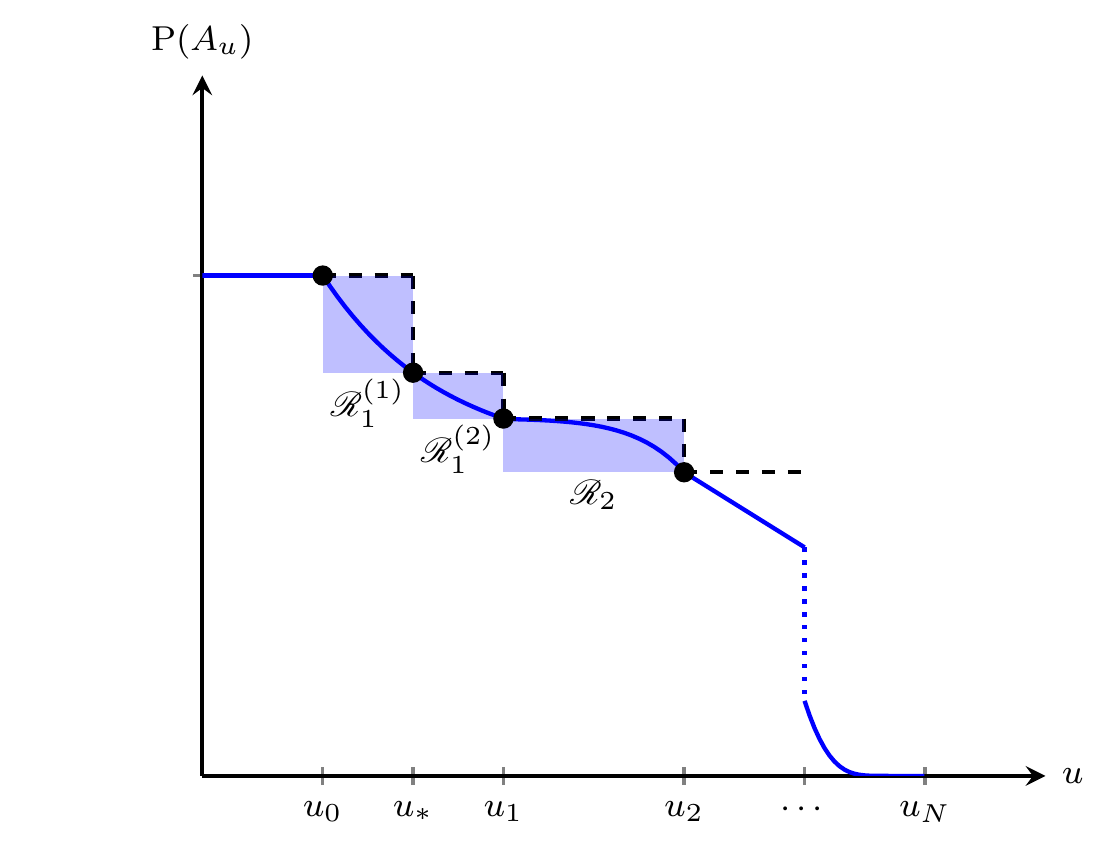}
\caption{}
\label{fig:approx-error2}
\end{subfigure}

\caption{(\subref{fig:approx-error1}) Step function $h^*(u)$ (dashed black lines) to approximate a concocted $\Prob(A_u)$ (solid blue curve) based on knots $u_0, \ldots, u_N$. Blue shaded areas represent rectangles $\mathcal{R}_j$. (\subref{fig:approx-error2}) Updated step function with $u_*$ inserted at midpoint of $u_0$ and $u_1$; $\mathcal{R}_1$ is replaced by $\mathcal{R}_1^{(1)}$ and $\mathcal{R}_1^{(2)}$.}
\label{fig:approx-error}
\end{figure*}

The step function $h^*(u)$ can be used to formulate a rejection sampler to take exact draws from $p(u)$ \citep{MartinoEtAl2018}. By construction, $h^*(u) \geq \Prob(A_{u})$ for all $u \in [0,1]$, so that rejection sampling can be carried out by Algorithm~\ref{alg:direct-sampling-with-rejection}. The bound in \eqref{eqn:rectangles} also bounds the probability of a rejection, which occurs when $[\text{Accept} = 0]$ in Line~\ref{step:dsr-acceptance} of Algorithm~\ref{alg:direct-sampling-with-rejection}.
\begin{proposition}
\label{result:direct-rejection}
The probability of rejection in Algorithm~\ref{alg:direct-sampling-with-rejection} is no greater than $\frac{c}{\psi} \sum_{j=1}^N |\mathcal{R}_j|$.
\end{proposition}

\begin{proof}
Suppose $V \sim \text{Uniform}(0,1)$ and $U \sim h(u)$ are independent random variables and let $M = a / (\psi / c)$. The probability of rejecting a candidate is
\begin{align*}
\Prob\left( V > \frac{\Prob(A_U)}{h^*(U)} \right)
&= \Prob\left( V > \frac{ p(U) }{ M h(U) } \right) \\
&= 1 - \E_U \left[ \Prob\left( V \leq \frac{ p(U) }{ M h(U) } \given U \right) \right] \\
&= 1 - \E_U \left[ \frac{ p(U) }{ M h(U) } \right]
= 1 - \frac{ 1 }{ M }
= \frac{a - \psi / c}{a}.
\end{align*}
Applying the inequality $a - \psi / c = \int_0^1 [h^*(u) - \Prob(A_u)] du \leq \sum_{j=1}^N |\mathcal{R}_j|$ from \eqref{eqn:unnorm-diff} to the numerator and $a \geq \psi / c$ from \eqref{eqn:norm-const-ineq} to the denominator gives the result.
\end{proof}

As anticipated, rejection is assured to be less likely when $h^*$ and $\Prob(A_u)$ are closer. A rejected $u$ may be added to the set of knot points to decrease the probability of a rejection in subsequent proposals, as shown in Line~\ref{step:dsr-update} of Algorithm~\ref{alg:direct-sampling-with-rejection}.

\begin{algorithm}
\caption{Rejection sampler based on direct sampling with step function $h$.}
\label{alg:direct-sampling-with-rejection}
\begin{algorithmic}[1]
\Do
\State Draw candidate $u$ from step density $h$.
\State Draw $v$ from $\text{Uniform}(0,1)$.
\State $\text{Accept} \leftarrow \ind\{ v \leq \Prob(A_u) / h^*(u) \}$.
\State Update $h^*$ with $u$ as additional knot if adaptive rejection is desired. \label{step:dsr-update}
\doWhile{$\text{Accept} = 0$} \label{step:dsr-acceptance}
\State Draw $x$ from $f(x \mid u)$.
\State \Return $x$.
\end{algorithmic}
\end{algorithm}

\begin{algorithm}
\caption{Select knots $u_1, \ldots, u_{N-1}$ to reduce $\sum_{j=1}^N |\mathcal{R}_j|$.}
\label{alg:small_rects}
\begin{algorithmic}
\State Let $u^{(0)} = u_L$, and $u^{(1)} = u_H$.
\For{$i = 1, \ldots, N-1$}
\State Let $u_0 < \ldots < u_{i}$ be sorted $u^{(0)}, \ldots, u^{(i)}$.
\State Let $|\mathcal{R}_j| = \{ \Prob(A_{u_{j-1}}) - \Prob(A_{u_{j}}) \} (u_{j} - u_{j-1})$ for $j = 1, \ldots, i$.
\State Let $j^* = \argmax\limits_{j=1, \ldots, i} \; |\mathcal{R}_j|$.
\State Let $u^{(i+1)} = \text{mid}(u_{j^*-1}, u_{j^*})$.
\EndFor
\State Let $u_0 < \ldots < u_{N}$ be sorted $u^{(0)}, \ldots, u^{(N)}$.
\State \Return $(u_0, \ldots, u_{N})$.
\end{algorithmic}
\end{algorithm}

\begin{remark}[Bisection method]
\label{remark:bisection}
A bisection search method \citep[e.g.][Section~5]{Lange2010} is useful in several computations in this section. Suppose $\mathcal{S} \subseteq \mathbb{R}$ and $\zeta(x) : \mathcal{S} \rightarrow \{ 0, 1 \}$ is a step function which increases from 0 to 1 at a point $x^*$. The objective of Algorithm~\ref{alg:bisection} is to identify $x^*$ by supplying lower and upper bounds $x_L < x_H$ such that $\zeta(x_L) = 0$ and $\zeta(x_H) = 1$, a function $\text{mid}(x,y) : \mathcal{S}^2 \rightarrow \mathcal{S}$ which returns a point in $[x,y]$, and a distance function $\text{dist}(x,y)$. We may therefore write $x^* = \min\{ x \in [x_L, x_H] : \zeta(x) = 1\}$. Algorithm~\ref{alg:bisection} is useful in the following computations.
\begin{enumerate}
\item To find $u_L$, the smallest $u \in [0,1]$ such that $\Prob(A_{u}) < \Prob(A_0)$, we first locate a sufficiently small $j^* \in \{ 0, 1, 2, 4, 8, \ldots \}$ until $\Prob(A_{\exp(-j)}) = \Prob(A_0)$. Algorithm~\ref{alg:bisection} may be used with $\zeta(u) = I\{ \Prob(A_u) < \Prob(A_0) \}$ with $x_L = e^{-j^*}$, and $x_H = 1$.

\item To find $u_H$, the smallest $u \in [0,1]$ such that $\Prob(A_{u}) = 0$, Algorithm~\ref{alg:bisection} may be used with $\zeta(u) = I\{ \Prob(A_{u}) > 0 \}$, $x_L = u_L$, and $x_H = 1$. 

\item The quantile function $H^{-1}(\varphi)$ may be evaluated by Algorithm~\ref{alg:bisection}. Given precomputed values $H(u_0)$, \ldots, $H(u_N)$ of the associated CDF, the index $\ell$ of the interval containing $\varphi$ can be identified using $\mathcal{S} = \{0, 1, \ldots, N \}$, $x_L = 0$, $x_H = N$, $\text{mid}(\ell_1, \ell_2) = \lfloor (\ell_1 + \ell_2) / 2 \rfloor$, and $\zeta(\ell) = I\{ H(u_\ell) \geq \varphi \}$. From here, linearity between $H(u_\ell)$ and $H(u_{\ell+1})$ yields
\begin{align*}
H^{-1}(\varphi) = u_\ell + (u_{\ell+1} - u_\ell) \{ \varphi - H(u_\ell) \} / \{ H(u_{\ell+1}) - H(u_\ell) \}.
\end{align*}
\end{enumerate}

\begin{algorithm}
\caption{Bisection search for $x^* =  \min\{ x \in [x_L, x_H] : \zeta(x) = 1 \}$. Inputs are bounds $x_L < x_H$, a step function $\zeta(x)$ with $\zeta(x_L) = 0$ and $\zeta(x_H) = 1$, a midpoint function $\text{mid}(x,y)$, a distance function $\text{dist}(x,y)$, and a tolerance $\delta > 0$.}
\label{alg:bisection}
\begin{algorithmic}
\State $x = \text{mid}(x_L, x_H)$
\While{$\text{dist}(x_L, x_H) > \delta$}
\State $x_L = \zeta(x) \cdot x_L + [1 - \zeta(x)] \cdot x$
\State $x_H = \zeta(x) \cdot x + [1 - \zeta(x)] \cdot x_H$
\State $x = \text{mid}(x_L, x_H)$
\EndWhile \\
\Return $x$
\end{algorithmic}
\end{algorithm}

\end{remark}

\section{Illustrative Examples}
\label{sec:examples}

We now demonstrate the direct sampler with step function through three examples. Algorithm~\ref{alg:direct-sampling-with-rejection} is used throughout with a prespecified number $N$ of initial knots selected by Algorithm~\ref{alg:small_rects}, and subsequent knots added through adaptive rejection. The direct sampler requires more computation than alternative methods which are mentioned in the three examples, but also generates an exact sample with relatively very few rejections. All reported run times were measured on an Intel Core i7--2600 3.40 GHz workstation with four CPU cores.

\subsection{Sampling from Conway-Maxwell Poisson}
\label{sec:example-cmp}

The Conway-Maxwell Poisson distribution has become popular in recent years as a count model which can express either over- or underdispersion \citep{ShmueliEtAl2005}. The Conway-Maxwell Poisson distribution $\text{CMP}(\lambda, \nu)$ has probability mass function (pmf)
\begin{align}
f(x \mid \btheta) = \frac{\lambda^x}{(x!)^\nu Z(\lambda, \nu)}, \quad
x = 0, 1, \ldots, \quad
\lambda > 0, \quad
\nu > 0,
\label{eqn:cmp-density}
\end{align}
a weighted density in the form of \eqref{eqn:weighted-density}, with normalizing constant $Z(\lambda, \nu) = \sum_{x=0}^\infty \lambda^x / (x!)^\nu$. The over- or underdispersion of $\text{CMP}(\lambda, \nu)$ is most readily compared to $\text{Poisson}(\lambda)$: CMP is overdispersed when $\nu < 1$, underdispersed when $\nu > 1$, and is equivalent when $\nu = 1$. CMP also has several other well-known special cases. When $\lambda \in (0,1)$ and $\nu = 0$, $\text{CMP}(\lambda, \nu)$ becomes a Geometric distribution with density $f(x \mid \lambda) = (1 - \lambda) \lambda^x$ for $x = 0, 1, \ldots$. When $\nu \rightarrow \infty$, the $\text{CMP}(\lambda, \nu)$ density converges to that of $\text{Bernoulli}(\lambda / (1 + \lambda))$.

The normalizing constant $Z(\lambda, \nu)$ is a series which does not appear to have a closed form. Approximating the normalizing constant has been a topic of interest \citep[e.g.][]{GauntEtAl2019}. The expansion
\begin{align*}
Z(\lambda, \nu) &= \frac{ \exp(\nu \lambda^{1/\nu}) }{ \lambda^{(\nu-1)/2\nu} (2\pi)^{(\nu-1)/2} \nu^{1/2} }
\left\{ 1 + O(\lambda^{-1/\nu}) \right\}
\end{align*}
given by \citet{ShmueliEtAl2005} illustrates in particular that the magnitude of $Z(\lambda, \nu)$ can vary wildly with $\lambda$ and $\nu$. For example, $Z(\lambda, 1) = e^{\lambda}$ for any $\lambda > 0$ but $Z(2, 0.075) \approx e^{780.515}$. Given this volatility, an exact method of generating variates which avoids computation of the normalizing constant is desirable.

Several recent papers have considered Bayesian analysis with CMP using the exchange algorithm \citep{MollerEtAl2006, MurrayEtAl2006}. The exchange algorithm utilizes a data augmentation step in Metropolis-Hastings sampling; an exact draw from the data-generating model is used to avoid computing the normalizing constant $Z(\lambda, \nu)$ in the acceptance ratio and therefore obtain an MCMC sampler for the unknown parameters. \citet{ChanialidisEtAl2018} and \citet{BensonFriel2021} take rejection sampling approaches to generate exact CMP draws and implement the exchange algorithm: \citet{ChanialidisEtAl2018} creates an envelope based on a piecewise Geometric distribution, while \citet{BensonFriel2021} use an envelope based on the Geometric distribution when $\nu < 1$ and Poisson otherwise. 

The direct sampler described in Algorithm~\ref{alg:direct-sampling-with-rejection} can be used to obtain exact draws from CMP with low probability of rejection and avoid explicit computation of the normalizing constant. To do this, the cases $\nu \geq 1$ and $\nu < 1$, corresponding to under- and overdispersion, are now addressed individually.

\paragraph{Case $\nu \geq 1$.}
Rewrite the unnormalized density \eqref{eqn:cmp-density} as
\begin{align}
f(x \mid \btheta)
\propto \frac{\lambda^x}{(x!)^\nu}
= \left(\frac{\lambda}{1+\lambda} \right)^x \frac{1}{1 + \lambda} (1 + \lambda)^{x+1}
\frac{\lambda^x}{(x!)^\nu}
\label{eqn:cmp-decomp1}
\end{align}
and let base density $g(x \mid \lambda) = [\lambda (1 + \lambda)^{-1}]^x  (1 + \lambda)^{-1}$, $x = 0, 1, \ldots$, be the pmf of a $\text{Geometric}(1 / \{1 + \lambda \})$ distribution. The weight function, on the log-scale, and its first and second derivative are respectively
\begin{align}
&\log w(x \mid \lambda, \nu) = (x+1) \log(1 + \lambda) - \nu \log \Gamma(x+1),
\nonumber \\
&\frac{\partial}{\partial x} \log w(x \mid \lambda, \nu)
= \log(1 + \lambda) - \nu \psi(x+1),
\label{eqn:cmp-weight-grad-underdisp} \\
&\frac{\partial^2}{\partial x^2} \log w(x \mid \lambda, \nu) = - \nu \psi'(x+1).
\label{eqn:cmp-weight-hess-underdisp}
\end{align}
Note that $\psi(x+1)$ is increasing for $x \geq 0$, where $\psi(1) = -0.57721566\ldots$ is the Euler–Mascheroni constant. For any $\lambda > 0$, $\log(1 + \lambda) > 0$ so that we may locate an $x_L$ and $x_H$ that yield negative and positive values of \eqref{eqn:cmp-weight-grad-underdisp}, respectively. A root $x^*$ of \eqref{eqn:cmp-weight-grad-underdisp} exists in $[x_L, x_H]$, and may be identified using a root-finding method such as Algorithm~\ref{alg:bisection}. The function \eqref{eqn:cmp-weight-hess-underdisp} is negative for all $x$, so that $\log w(x \mid \lambda, \nu)$ is concave and $x^*$ is a maximizer with $c = \log w(x^* \mid \lambda, \nu)$. To find the endpoints $x_1(u) < x_2(u)$ of the interval $A_u = \{ x > 0 : \log w(x \mid \lambda, \nu) > \log (uc) \}$, root-finding may be applied twice to the function $\log w(x \mid \lambda, \nu) - \log (uc)$: once to obtain $x_1(u)$ from the interval $[0, c]$, and again to obtain $x_2(u)$ from the interval $[c, x_H^*]$, where $x_H^*$ is a number large enough that $\log w(x_H^* \mid \lambda, \nu) - \log (uc)$ is negative.

\paragraph{Case $\nu < 1$.}
Variates from $\text{CMP}(\lambda, \nu)$ can become very large as $\nu$ is taken closer to zero, especially when $\lambda \geq 1$. Here, the support of a $\text{Geometric}(1 / \{1 + \lambda \})$ base distribution may be practically disjoint from the target CMP, leading to extremely small probabilities in computations such as \eqref{eqn:truncated-cdf} and \eqref{eqn:truncated-quantile}. To illustrate, suppose $X \sim \text{CMP}(\lambda, \nu)$ with $\lambda = 2$ and $\nu = 0.075$; here, $\Prob(X \leq 7,306) \approx e^{-40}$ but $\Prob(S > 7,086) \approx e^{-2873.531}$ for $S \sim \text{Geometric}(1 / \{1 + \lambda \})$ so that $[X \leq S]$ effectively never occurs. A more convenient base distribution is given by the reparameterization of CMP based on $\nu$ and $\mu = \lambda^{1/\nu}$ used by \citet{GuikemaGoffelt2008} to formulate regression models. The unnormalized portion of density \eqref{eqn:cmp-density} may now be decomposed as
\begin{align}
f(x \mid \btheta)
\propto \frac{\mu^{\nu x}}{(x!)^\nu}
= \left(\frac{\mu}{1+\mu} \right)^x \frac{1}{1+\mu} (1 + \mu)^{x+1}
\frac{\mu^{x(\nu-1)}}{(x!)^\nu}
\label{eqn:cmp-decomp2}
\end{align}
so that the base density $g(x \mid \mu) = [\mu (1 + \mu)^{-1}]^x  (1 + \mu)^{-1}$, $x = 0, 1, \ldots$, is the pmf of $T \sim \text{Geometric}(1 / \{1 + \mu\})$. The quantiles of $T$ corresponding to probabilities 0.025 and 0.975 are 261 and 38,075, compared to 9,607 and 11,061 for $X$, suggesting that the distribution of $T$ is more suitable than that of $S$ as a base distribution. The log of the weight function and its first and second derivative are now, respectively,
\begin{eqnarray}
&\log w(x \mid \mu, \nu) &= (x+1) \log(1 + \mu) - \nu \log \Gamma(x+1) \nonumber \\
& &\quad + x(\nu-1)\log \mu,
\nonumber \\
&\frac{\partial}{\partial x} \log w(x \mid \mu, \nu)
&= \left\{ \log(1 + \mu) + (\nu-1) \log \mu \right\} \label{eqn:cmp-weight-grad-overdisp} \\
& &\quad - \nu \psi(x+1),
\nonumber \\
& \frac{\partial^2}{\partial x^2} \log w(x \mid \mu, \nu) &= - \nu \psi'(x+1).
\label{eqn:cmp-weight-hess-overdisp}
\end{eqnarray}
A root of \eqref{eqn:cmp-weight-grad-overdisp} exists if the term $\log(1 + \mu) + (\nu-1) \log \mu$ is positive. To verify positivity, if $\mu < 1$, then
\(
\log(1 + \mu) + (\nu-1) \log \mu \geq \log(1 + \mu) \geq 0;
\)
On the other hand, if $\mu \geq 1$ then
\(
\log(1 + \mu) + (\nu-1) \log \mu \geq \log(1 + \mu) - \log \mu = \log(1/\mu + 1) \geq 0.
\)
Maximization and root-finding may then proceed similarly to the case where $\nu \geq 1$.

\vspace{1em}

Figure~\ref{fig:cmp_step} compares the unnormalized density $\Prob(A_u)$ with unnormalized step function $h^*(u)$ in two CMP settings with $\lambda = 2$: one with $\nu = 0.5$ using $N = 13$ knots and one with $\nu = 0.2$ using $N = 20$ knots, corresponding to progressively higher levels of overdispersion. Three different knot selection methods are shown for comparison: equal spacing, Algorithm~\ref{alg:small_rects} using geometric midpoints, and Algorithm~\ref{alg:small_rects} using arithmetic midpoints. Although $\nu < 1$ in this example, the step function has been constructed from decomposition \eqref{eqn:cmp-decomp1} to illustrate the effect of using a base distribution which differs either moderately and greatly from the target. The case $\nu = 0.5$ is handled relatively well by all three methods, with arithmetic midpoint providing the best approximation followed by equal spacing, then geometric midpoint. A decrease in $\nu$ to $0.2$ is seen to create a much more difficult situation, with much of the density occurring in a small subinterval of $[u_L, u_H]$. Here, equal spacing will require a large $N$ to obtain a useful approximation. Algorithm~\ref{alg:small_rects} with arithmetic midpoint produces a better approximation than equal spacing, but has not yet located the steep descent shown on the left of the display. On the other hand, the geometric midpoint is able to capture this feature; this is due to its suitability with very small magnitude numbers as discussed in Remark~\ref{remark:midpoint}.

\begin{remark}[Weighted rectangles]
The approximation in Figures~\ref{fig:cmp_case2_step2} and \ref{fig:cmp_case2_step3} can be further improved, without increasing $N$, by using a weighted priority
\(
\omega \log \{ \Prob(A_{u_{j-1}}) - \Prob(A_{u_{j}}) \} + (1 - \omega) \log (u_{j} - u_{j-1}),
\)
$\omega \in (0,1)$, in place of $\log |\mathcal{R}_j|$ to order the rectangles. In particular, $\omega > 1/2$ prioritizes taller rectangles over wider ones having equal area which encourages knot placement at sudden descents occurring on very short intervals.
\end{remark}

Figure~\ref{fig:cmp-draws} displays draws of $\text{CMP}(\lambda, \nu)$ from Algorithm~\ref{alg:direct-sampling-with-rejection} with $\lambda = 2$ and $\nu \in \{ 0.05, 0.5, 2, 5 \}$. Here, decomposition \eqref{eqn:cmp-decomp1} is used with $\nu \in \{ 2, 5 \}$ and \eqref{eqn:cmp-decomp2} is used with $\nu \in \{ 0.05, 0.5 \}$. As anticipated, the empirical pmf of 20,000 draws matches closely to the exact pmf \eqref{eqn:cmp-density}. With $N = 10$ knots initially selected in each case, the number of rejections was 279, 86, 40, and 27 in Figures~\ref{fig:cmp-draws2}, \ref{fig:cmp-draws1}, \ref{fig:cmp-draws3}, and \ref{fig:cmp-draws4}, respectively. This demonstrates the ability of the samplers obtained in this section to generate CMP variates with small probability of rejection. This may be contrasted to acceptance rates as low as about 20\% reported by \citet{BensonFriel2021}; however, their rejection sampler requires less computation and therefore may be faster in practice.

\begin{figure*}
\centering
\begin{subfigure}{0.32\textwidth}
\includegraphics[width=\textwidth]{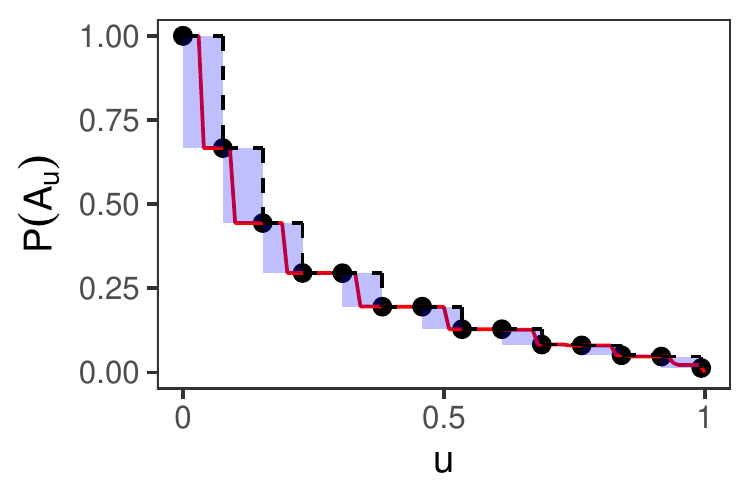}
\caption{$\sum_{j=1}^N |\mathcal{R}_j| = 0.0754$.}
\label{fig:cmp_case1_step1}
\end{subfigure}
\begin{subfigure}{0.32\textwidth}
\includegraphics[width=\textwidth]{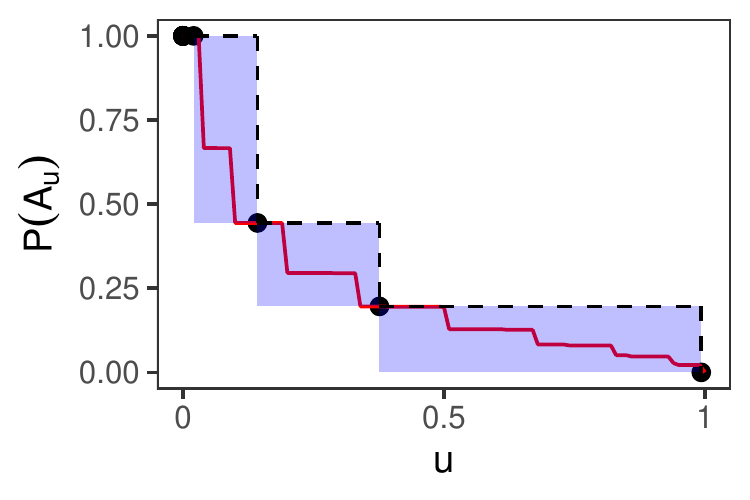}
\caption{$\sum_{j=1}^N |\mathcal{R}_j| = 0.2468$.}
\label{fig:cmp_case1_step2}
\end{subfigure}
\begin{subfigure}{0.32\textwidth}
\includegraphics[width=\textwidth]{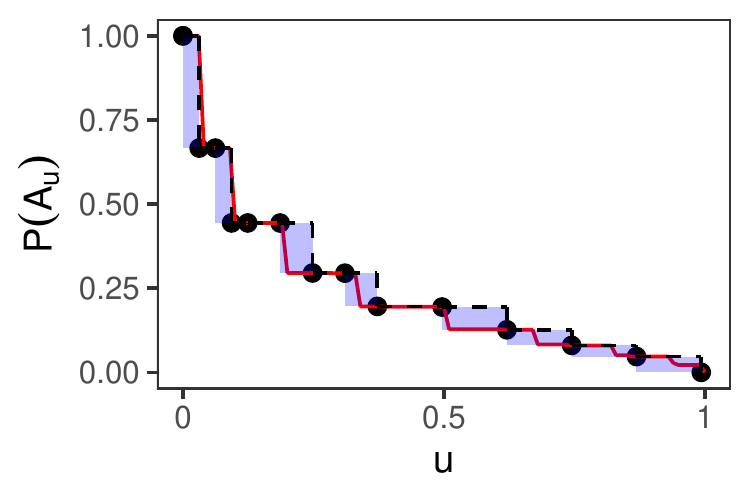}
\caption{$\sum_{j=1}^N |\mathcal{R}_j| = 0.0570$.}
\label{fig:cmp_case1_step3}
\end{subfigure}

\begin{subfigure}{0.32\textwidth}
\includegraphics[width=\textwidth]{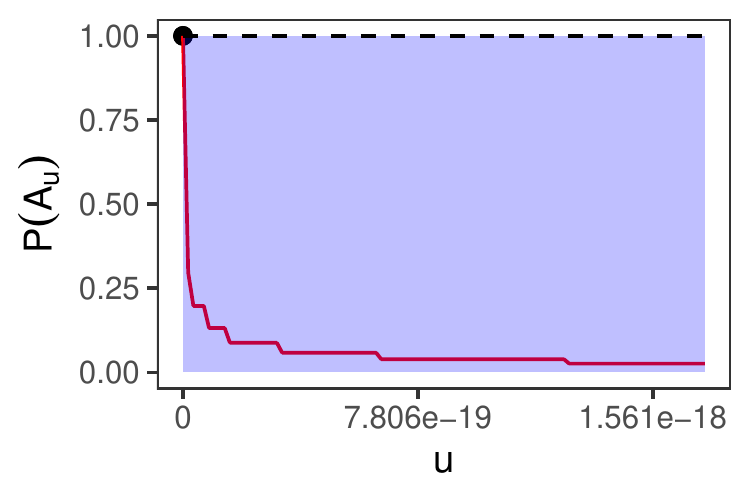}
\caption{$\sum_{j=1}^N |\mathcal{R}_j| = 4.561 \times 10^{-11}$.}
\label{fig:cmp_case2_step1}
\end{subfigure}
\begin{subfigure}{0.32\textwidth}
\includegraphics[width=\textwidth]{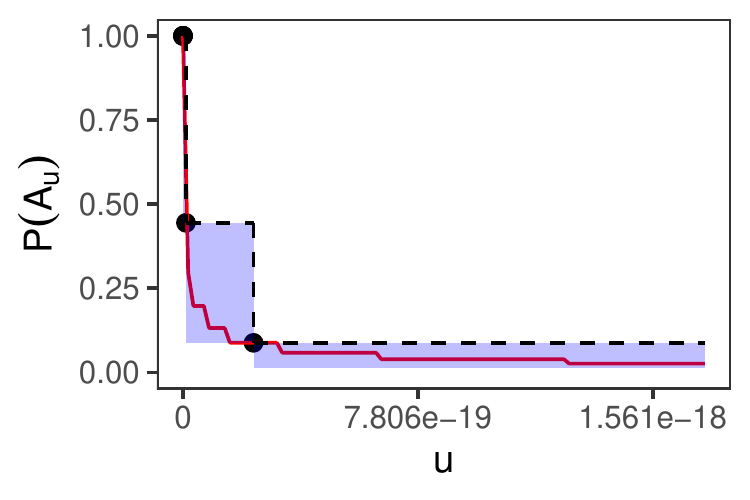}
\caption{$\sum_{j=1}^N |\mathcal{R}_j| = 1.007 \times 10^{-17}$.}
\label{fig:cmp_case2_step2}
\end{subfigure}
\begin{subfigure}{0.32\textwidth}
\includegraphics[width=\textwidth]{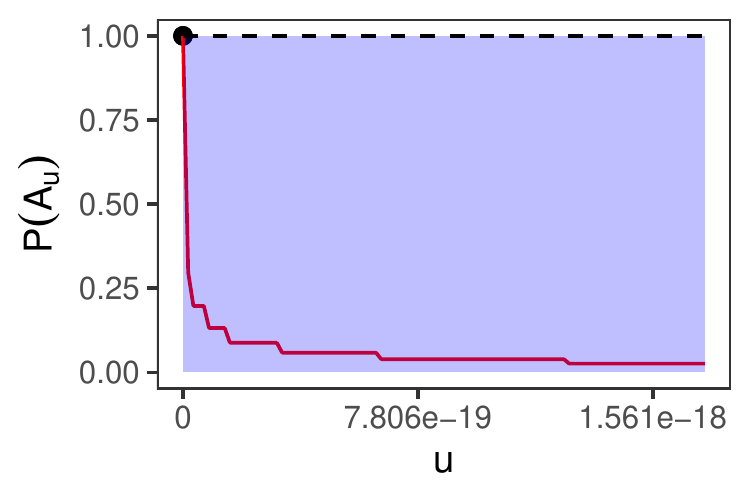}
\caption{$\sum_{j=1}^N |\mathcal{R}_j| = 1.743 \times 10^{-15}$.}
\label{fig:cmp_case2_step3}
\end{subfigure}

\caption{Realizations of $\Prob(A_u)$ (solid red line) and $h^*(u)$ (dashed line) for Conway-Maxwell Poisson with $\lambda = 2$. Shaded rectangles represent the rectangles $\mathcal{R}_j$. The top row (\subref{fig:cmp_case1_step1}), (\subref{fig:cmp_case1_step2}), and (\subref{fig:cmp_case1_step3}) correspond to $\nu = 0.5$ with $N = 13$ while the bottom row (\subref{fig:cmp_case2_step1}), (\subref{fig:cmp_case2_step2}), and (\subref{fig:cmp_case2_step3}) correspond to $\nu = 0.2$ with $N = 20$. Knots for (\subref{fig:cmp_case1_step1}) and (\subref{fig:cmp_case2_step1}) were selected by equal spacing, (\subref{fig:cmp_case1_step2}) and (\subref{fig:cmp_case2_step2}) were selected using Algorithm~\ref{alg:small_rects} with geometric mean to compute midpoints, and (\subref{fig:cmp_case1_step3}) and (\subref{fig:cmp_case2_step3}) were selected using Algorithm~\ref{alg:small_rects} with arithmetic mean to compute midpoints. Knots are excluded from the displays in (\subref{fig:cmp_case2_step1}), (\subref{fig:cmp_case2_step2}), and (\subref{fig:cmp_case2_step3}). Subcaptions display total rectangle area $\sum_{j=1}^N |\mathcal{R}_j|$ achieved by the approximation.}
\label{fig:cmp_step}
\end{figure*}

\begin{figure*}
\centering
\begin{subfigure}{0.24\textwidth}
\includegraphics[width=\textwidth]{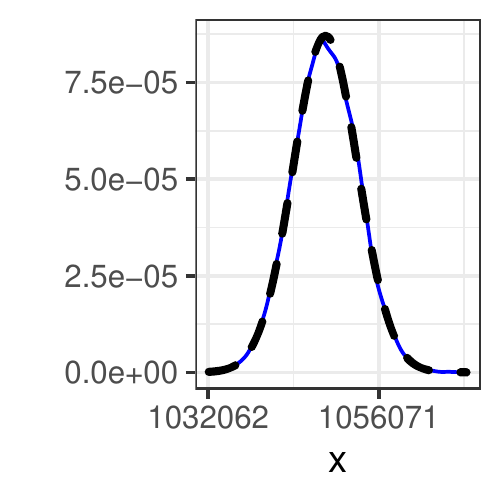}
\caption{$\nu = 0.05$.}
\label{fig:cmp-draws2}
\end{subfigure}
\begin{subfigure}{0.24\textwidth}
\includegraphics[width=\textwidth]{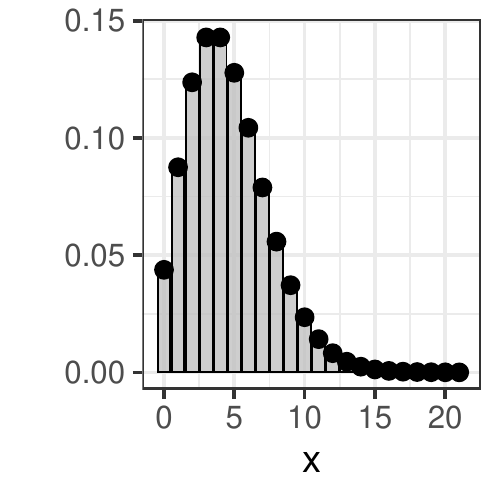}
\caption{$\nu = 0.5$.}
\label{fig:cmp-draws1}
\end{subfigure}
\begin{subfigure}{0.24\textwidth}
\includegraphics[width=\textwidth]{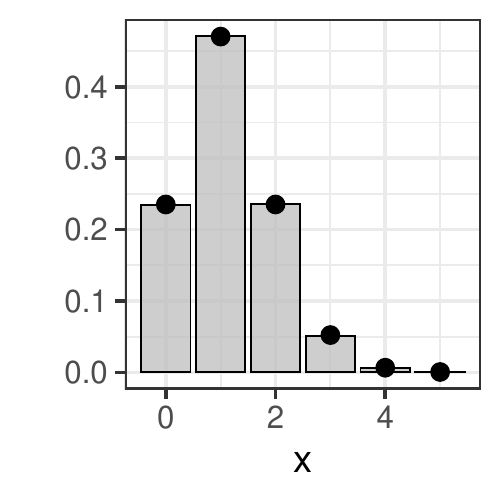}
\caption{$\nu = 2$.}
\label{fig:cmp-draws3}
\end{subfigure}
\begin{subfigure}{0.24\textwidth}
\includegraphics[width=\textwidth]{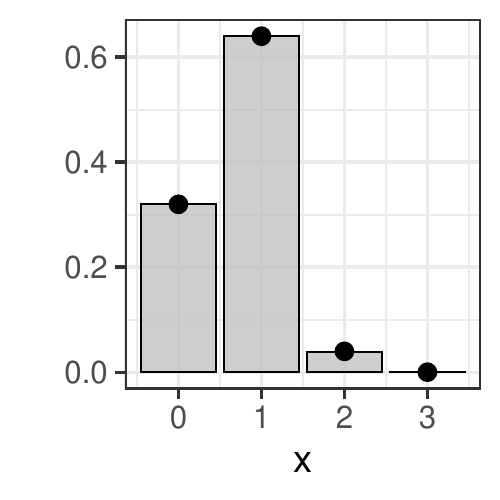}
\caption{$\nu = 5$.}
\label{fig:cmp-draws4}
\end{subfigure}
\caption{Empirical density of 20,000 draws versus pmf of CMP from with $\lambda = 2$ and $\nu$ as specified in the subcaption. In (\subref{fig:cmp-draws2}), solid blue line and dashed black line represent empirical density and CMP pmf, respectively. In (\subref{fig:cmp-draws1}), (\subref{fig:cmp-draws3}), and (\subref{fig:cmp-draws4}), gray bars and black dots represent empirical density and CMP pmf, respectively.}
\label{fig:cmp-draws}
\end{figure*}

\subsection{Sampling the Dependence Parameter in Conditional Autoregression}
\label{sec:example-car}
In a conditional autoregression (CAR) setting \citep[e.g.][Section~6.6]{Cressie1991}, the joint distribution of a random vector $\vec{\eta} = (\eta_1, \ldots, \eta_k)$ implies a certain regression for each conditional distribution $[\eta_i \mid \vec{\eta}_{-i}]$, $i = 1, \ldots, k$, where $\vec{\eta}_{-i} = (\eta_1, \ldots, \eta_{i-1}, \eta_{i+1}, \ldots \eta_k)$. Let us consider a particular mixed effects CAR model which is useful for data observed on areal units in a spatial domain. Suppose there are $k$ distinct areas and let $\vec{A} = (a_{ij})$ be a $k \times k$ adjacency matrix; $a_{ij} = 1$ if areas $i$ and $j$ are adjacent and $i \neq j$, otherwise $a_{ij} = 0$. Let $\vec{D} = \Diag(a_{1+}, \ldots, a_{k+})$ with $a_{i+} = \sum_{j=1}^k a_{ij}$ be a diagonal matrix containing the row sums of $\vec{A}$. Suppose $\vec{y} = (y_1, \ldots, y_n)$ is a vector of observed outcomes, $\vec{X} \in \mathbb{R}^{n \times d}$ and  $\vec{S} \in \mathbb{R}^{n \times k}$ are fixed design matrices, and
\begin{align}
&\vec{y} = \vec{X} \vec{\beta} + \vec{S} \vec{\eta} + \vec{\epsilon}, \quad
\vec{\epsilon} \sim \text{N}(\vec{0}, \sigma^2 \vec{I}), \label{eqn:car-model} \\
&\vec{\eta} \sim \text{N}(\vec{0}, \tau^2 (\vec{D} - \rho \vec{A})^{-1}).
\nonumber
\end{align}
Here it can be shown that the conditionals have distribution
\(
[\eta_i \mid \vec{\eta}_{-i}] \sim \text{N}(
(\rho / a_{i+}) \sum_{j=1}^k \eta_j a_{ij},
\tau^2 / a_{i+}).
\)
The parameter $\rho$ must be in the interval $[0,1]$; the matrix $\vec{D} - \rho \vec{A}$ is nonsingular provided that $\rho < 1$, while the inverse does not exist when $\rho = 1$ and a pseudo-inverse may be instead considered. To complete a Bayesian specification of the model, consider the prior
\begin{align}
&\vec{\beta} \sim \text{N}(0, \sigma_\beta^2 \vec{I}) \quad
\sigma^2 \sim \text{Uniform}(0, M_\sigma) \label{eqn:car-prior} \\
&\tau^2 \sim \text{Uniform}(0, M_\tau) \quad
\rho \sim \text{Uniform}(0, 1),
\nonumber
\end{align}
following \citet{Lee2013}. From \eqref{eqn:car-model} and \eqref{eqn:car-prior}, and regarding $\vec{\eta}$ as augmented data to be drawn with the parameters $\btheta = (\vec{\beta}, \sigma^2, \tau^2)$ and $\rho$, the following conditionals with familiar distributions are obtained for a Gibbs sampler:
\begin{enumerate}
\item $[\vec{\beta} \mid \rest] \sim \text{N}(\vec{\vartheta}_\beta, \vec{\Omega}_\beta^{-1})$ where
\(
\vec{\Omega}_\beta = \sigma^{-2} \vec{X}^\top \vec{X} + \sigma_\beta^{-2} \vec{I}
\)
and
\(
\vec{\vartheta}_\beta = \sigma^{-2} \vec{\Omega}_\beta^{-1} \vec{X}^\top (\vec{y} - \vec{S} \vec{\eta}),
\)

\item $[\vec{\eta} \mid \rest] \sim \text{N}(\vec{\vartheta}_\eta, \vec{\Omega}_\eta^{-1})$ where
\(
\vec{\Omega}_\eta = \sigma^{-2} \vec{S}^\top \vec{S} + \tau^{-2} (\vec{D} - \rho \vec{A})
\)
and
\(
\vec{\vartheta}_\eta = \sigma^{-2} \vec{\Omega}_\eta^{-1} \vec{S}^\top (\vec{y} - \vec{X} \vec{\beta}),
\)

\item $[\sigma^2 \mid \rest] \sim \text{IG}_{[0, M_{\sigma}]}(a_\sigma, b_\sigma)$, an Inverse Gamma distribution with shape $a_\sigma = n/2$ and rate
\(
b_\sigma = \frac{1}{2} \lVert \vec{y} - \vec{X} \vec{\beta} - \vec{S} \vec{\eta}  \rVert^2,
\)

\item $[\tau^2 \mid \rest] \sim \text{IG}_{[0, M_{\tau}]}(a_\tau, b_\tau)$ with $a_\tau = k/2$ and
\(
b_\tau = \frac{1}{2} \vec{\eta}^\top (\vec{D} - \rho \vec{A}) \vec{\eta}.
\)
\end{enumerate}
Here, $[\vec{X} \mid \rest]$ denotes the distribution of $\vec{X}$ based on all other random variables and a distribution with subscript $[a,b]$ denotes that it is truncated to that interval. The conditional for $\rho$ takes the more unfamiliar form
\begin{align}
f(\rho \mid \rest) &\propto |\vec{D} - \rho \vec{A}|^{1/2}
\exp\left\{ -\frac{\rho}{2 \tau^2} \vec{\eta}^\top \vec{A} \vec{\eta} \right\} \ind\{\rho \in [0,1]\}.
\label{eqn:rho-cond}
\end{align}
\citet{Lee2013} uses a Metropolis-Hastings approach to sample from \eqref{eqn:rho-cond}. At the $r$th iteration, a candidate $\rho^*$ is drawn from truncated Normal proposal distribution $\text{N}_{[0,1]}(\rho^{(r-1)}, \sigma_\text{prop}^2)$ so that $\rho^{(r)}$ is assigned to $\rho^*$ with probability
\(
\min\{1, f(\rho^* \mid \rest) / f(\rho^{(r-1)} \mid \rest) \}
\)
and to $\rho^{(r-1)}$ otherwise. This requires selecting---or adaptively tuning---the proposal variance $\sigma_\text{prop}^2$ to be large enough that the chain is not restricted to very small moves, but not too large that many proposals are rejected. Let us now consider a direct sampler to generate exact draws from \eqref{eqn:rho-cond}. First, suppose $\vec{U} \vec{\Phi} \vec{U}^\top$ is the spectral decomposition of $\vec{A}$ and let $\vec{Q} = \vec{\Phi}^{1/2} \vec{U}^\top \vec{D}^{-1} \vec{U} \vec{\Phi}^{1/2}$. Using a well-known property of determinants \citep[e.g.][Theorem~10.11]{BanerjeeRoy2014},
\begin{align}
|\vec{D} - \rho \vec{A}|
&= |\vec{D} - \rho \vec{U} \vec{\Phi} \vec{U}^\top| \nonumber \\
&= |\vec{D}| \cdot |\vec{I} - \rho \vec{Q}|.
\label{eqn:det-swap}
\end{align}
Let $\lambda_1 \geq \cdots \geq \lambda_k$ be the eigenvalues of $\vec{Q}$ with corresponding eigenvectors $\vec{v}_1, \ldots, \vec{v}_k$. Then $\vec{v}_i$ is also an eigenvector of $\vec{I} - \rho \vec{Q}$ with corresponding eigenvalue $1 - \rho \lambda_i$. Therefore, $|\vec{D} - \rho \vec{A}| = \prod_{i=1}^k a_{i+} \cdot \prod_{i=1}^k (1 - \rho \lambda_i)$. Note that the elements of $\vec{\Phi}^{1/2}$ may be complex numbers but $\lambda_1, \ldots, \lambda_k$ are real. From \eqref{eqn:rho-cond}, we may write
\begin{align*}
f(\rho \mid \rest) \propto \left[ \prod_{i=1}^k (1 - \rho \lambda_i) \right]^{1/2}
\exp\left\{ \frac{\rho}{2 \tau^2} \vec{\eta}^\top \vec{A} \vec{\eta} \right\} \ind\{\rho \in [0,1]\}.
\end{align*}
Let us take $g(\rho) = \ind\{\rho \in [0,1]\}$ so that the base distribution is $\text{Uniform}(0,1)$ and the weight function $w$ is specified on the log-scale by
\begin{align}
\log w(\rho) = \frac{1}{2} \sum_{i=1}^k \log(1 - \rho \lambda_i) + \frac{\rho}{2 \tau^2} \vec{\eta}^\top \vec{A} \vec{\eta} + \log \ind\{\rho \in [0,1]\}.
\label{eqn:car-weight-fn}
\end{align}
For $\rho \in (0,1)$,
\begin{align}
&\frac{\partial}{\partial \rho} \log w(\rho) = -\frac{1}{2} \sum_{i=1}^k \frac{\lambda_i}{1 - \rho \lambda_i} + \frac{1}{2 \tau^2} \vec{\eta}^\top \vec{A} \vec{\eta},
\label{eqn:car-weight-grad} \\
&\frac{\partial^2}{\partial \rho^2} \log w(\rho) = -\frac{1}{2} \sum_{i=1}^k \frac{\lambda_i^2}{(1 - \rho \lambda_i)^2}.
\label{eqn:car-weight-hess}
\end{align}
Now, for $\rho \in [0,1)$ and assuming all areas have at least one adjacent neighbor so that $\vec{D} - \rho \vec{A}$ is positive definite, both $|\vec{D}| = \prod_{i=1}^k a_{k+} > 0$ and $|\vec{D} - \rho \vec{A}| > 0$. Therefore, \eqref{eqn:det-swap} implies $0 < |\vec{I} - \rho \vec{Q}| = \prod_{i=1}^k (1 - \rho \lambda_i)$ so that $\rho \lambda_i < 1$ for each $i = 1, \ldots, k$. Now it can be seen that \eqref{eqn:car-weight-hess} is negative and a root of \eqref{eqn:car-weight-grad} is a maximum of \eqref{eqn:car-weight-fn}. Furthermore, $\frac{1}{2} \sum_{i=1}^k \frac{\lambda_i}{1 - \rho \lambda_i}$ is an increasing function of $\rho$ so that \eqref{eqn:car-weight-grad} has at most one root. Therefore, the maximizer $\rho^*$ of \eqref{eqn:car-weight-fn} occurs at the root if it exists; otherwise, it occurs at one of the endpoints $\{0,1\}$ of the domain. To find the roots $\{ \rho_1(u), \rho_2(u)\}$ of the interval $A(u) = \{ \rho \in [0,1] : w(\rho) > u c \}$, note that $\rho_1(u) = 0$ if $\log w(0) > \log(uc)$; otherwise, a solution $\rho_1(u)$ to $\log w(\rho) = \log(uc)$ may be found in $[0, \rho^*]$ numerically. Similarly, $\rho_2(u) = 1$ if $\log w(1) > \log(uc)$; otherwise, a solution $\rho_2(u)$ to $\log w(\rho) = \log(uc)$ may be found in $[\rho^*, 1]$ numerically. Operations involving the $\text{Uniform}(0,1)$ base distribution outlined in Section~\ref{sec:truncated-base} are simple, using expressions for the CDF $G(\rho) = \rho$ for $\rho \in [0,1]$ and quantile function $G^{-}(\varphi) = \varphi$ for $\varphi \in [0,1]$.

Now we have a complete Gibbs sampler based on conjugate steps to draw $\vec{\beta}$, $\vec{\eta}$, $\sigma^2$, and $\tau^2$, and a direct sampling step to draw $\rho$. To illustrate the sampler, we revisit the analysis from \citet{Lee2013} on property prices in Glasgow, Scotland. The data are available in the CARBayesdata package \citep{CARBayesdata2020}. There are $k = 271$ areal units with one observation per area so that $n = k$. The response $\vec{y}$ is taken to be log of median housing price (in thousands) of properties sold in 2008. Columns of design matrix $\vec{X}$ include an intercept (corresponding to $\beta_0$), log of number of recorded crimes per 10,000 residents ($\beta_1$), median number of rooms in a property ($\beta_2$), percentage of properties which sold in a year ($\beta_3$), and log of average driving time to the nearest shopping center ($\beta_7$). The remaining columns are based on a categorical variable indicating the most prevalent property type in the area with levels: ``flat'' ($\beta_4$), ``semi-detached'' ($\beta_5$), ``terraced'' ($\beta_6$), and ``detached'' (baseline). Because there is one observation per area, $\vec{S}$ is taken to be a $k \times k$ identity matrix.

Following \citet{Lee2013}, hyperparameter values are taken to be $\sigma_\beta^2 = 1000$, $M_\sigma = 1000$, and $M_\tau = 1000$. A direct sampler is used to draw $\rho$ following Algorithm~\ref{alg:direct-sampling-with-rejection}. Initially, $N = 30$ knots are selected in each iteration of the Gibbs sampler via Algorithm~\ref{alg:small_rects}. Table~\ref{tab:car-fit} compares summaries of the draws from this Gibbs sampler to those from CARBayes version 1.6 \citep{Lee2013}, utilizing a Gibbs sampler with Metropolis-Hastings step for $\rho$. Figure~\ref{fig:car-rho-mcmc} displays draws of $\rho$ from the two samplers. Following \citet{Lee2013}, results for both samplers are based on a chain of 100,000 draws with 20,000 discarded as burn-in and keeping one of every remaining 10 to yield 8,000 draws from each. The two results are quite similar, the most notable difference being that the posterior of $\rho$ is somewhat more right-skewed under the direct sampler. To obtain 100,000 draws of $\rho$, Metropolis-Hastings step rejected 40,232 proposals while the direct sampler rejected a total of only 458 proposals. Recall that each draw of Metropolis-Hastings samples approximately from conditional \eqref{eqn:rho-cond} while direct sampling with rejection is exact. However, the direct sampler required substantially more computation ub this setting, taking on the order of 1.7 hours compared to 10 minutes for CARBayes. Performance improvements for the former may be possible; in particular, the weight function \eqref{eqn:car-weight-fn} changes only by an additive constant $\frac{\rho}{2 \tau^2} \vec{\eta}^\top \vec{A} \vec{\eta}$ within each step of the Gibbs sampler so that repetition of some computations may be avoided.

\begin{table*}
\centering
\caption{Summary of draws for each parameter based on 8,000 saved draws: 100,000 total draws with 20,000 discarded as burn-in and 9 of every remaining 10 discarded for thinning.}
\label{tab:car-fit}
\begin{subtable}{0.48\textwidth}
\centering
\caption{Gibbs with direct sampling step.}
\label{tab:car-fit-direct}
\begin{tabular}{crrrr}
\toprule
\multicolumn{1}{c}{} &
\multicolumn{1}{c}{Mean} &
\multicolumn{1}{c}{SD} &
\multicolumn{1}{c}{2.5\%} &
\multicolumn{1}{c}{97.5\%} \\
\cmidrule(lr){1-1}
\cmidrule(lr){2-5}
$\beta_0$  &  4.7745 & 0.2537 &  4.2767 &  5.2608 \\
$\beta_1$  & -0.1129 & 0.0305 & -0.1721 & -0.0531 \\
$\beta_2$  &  0.2218 & 0.0253 &  0.1727 &  0.2720 \\
$\beta_3$  &  0.0023 & 0.0003 &  0.0017 &  0.0029 \\
$\beta_4$  & -0.2533 & 0.0578 & -0.3677 & -0.1417 \\
$\beta_5$  & -0.1624 & 0.0500 & -0.2602 & -0.0647 \\
$\beta_6$  & -0.2901 & 0.0627 & -0.4153 & -0.1671 \\
$\beta_7$  & -0.0017 & 0.0289 & -0.0581 &  0.0552 \\
$\sigma^2$ &  0.0244 & 0.0047 &  0.0151 &  0.0334 \\
$\tau^2$   &  0.0479 & 0.0180 &  0.0208 &  0.0903 \\
$\rho$     &  0.9885 & 0.0109 &  0.9591 &  0.9992 \\
\bottomrule
\end{tabular}
\end{subtable}
\begin{subtable}{0.48\textwidth}
\centering
\caption{Gibbs with Metropolis-Hastings step.}
\label{tab:car-fit-mh}
\begin{tabular}{crrrr}
\toprule
\multicolumn{1}{c}{} &
\multicolumn{1}{c}{Mean} &
\multicolumn{1}{c}{SD} &
\multicolumn{1}{c}{2.5\%} &
\multicolumn{1}{c}{97.5\%} \\
\cmidrule(lr){1-1}
\cmidrule(lr){2-5}
$\beta_0$  &  4.7576 & 0.2459 &  4.2734 &  5.2394 \\
$\beta_1$  & -0.1116 & 0.0307 & -0.1731 & -0.0525 \\
$\beta_2$  &  0.2222 & 0.0261 &  0.1710 &  0.2733 \\
$\beta_3$  &  0.0023 & 0.0003 &  0.0016 &  0.0029 \\
$\beta_4$  & -0.2558 & 0.0581 & -0.3687 & -0.1419 \\
$\beta_5$  & -0.1637 & 0.0512 & -0.2638 & -0.0635 \\
$\beta_6$  & -0.2927 & 0.0628 & -0.4157 & -0.1727 \\
$\beta_7$  & -0.0017 & 0.0294 & -0.0594 &  0.0556 \\
$\sigma^2$ &  0.0237 & 0.0048 &  0.0144 &  0.0333 \\
$\tau^2$   &  0.0540 & 0.0192 &  0.0241 &  0.0995 \\
$\rho$     &  0.9815 & 0.0146 &  0.9429 &  0.9979 \\
\bottomrule
\end{tabular}
\end{subtable}
\end{table*}

\begin{figure*}
\centering
\begin{subfigure}{0.48\textwidth}
\centering
\includegraphics[width=\textwidth]{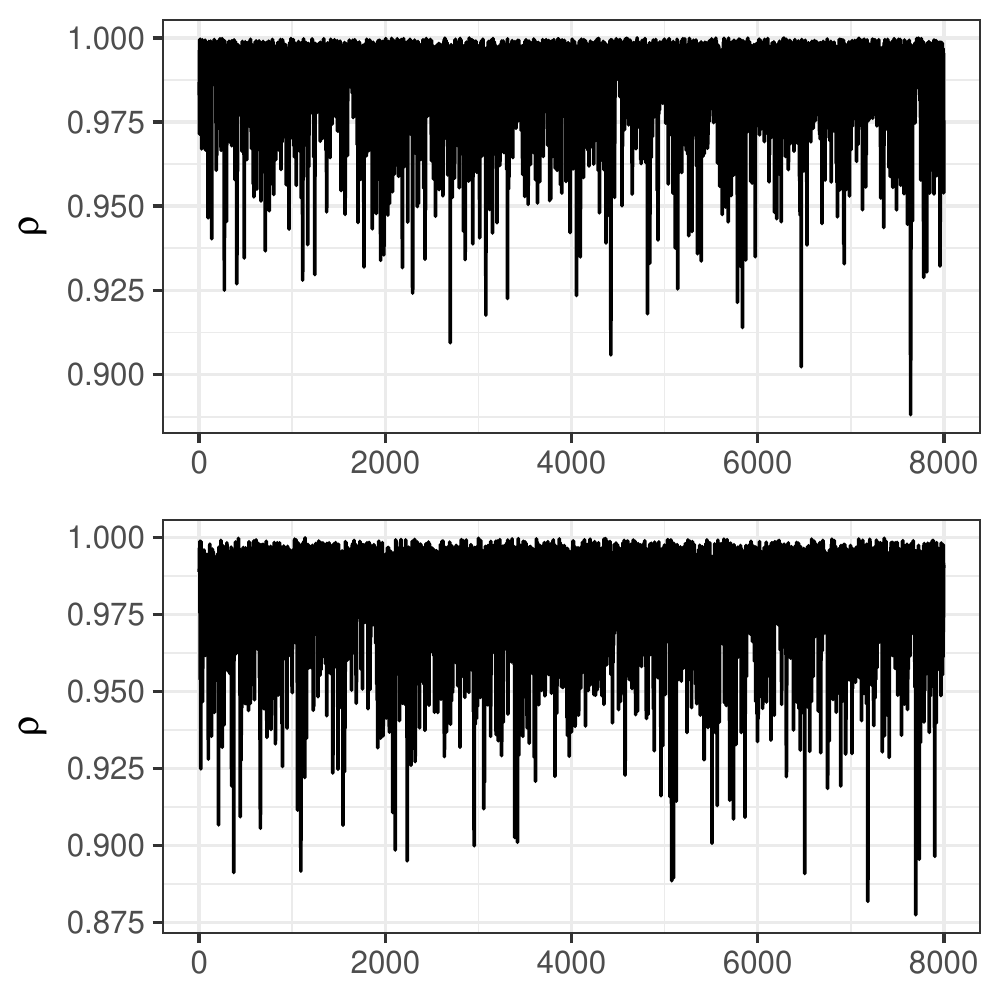}
\caption{Trace plots from direct sampler (top) and Metropolis-Hastings (bottom).}
\label{fig:car-rho-mcmc-draws}
\end{subfigure}
\begin{subfigure}{0.48\textwidth}
\centering
\includegraphics[width=\textwidth]{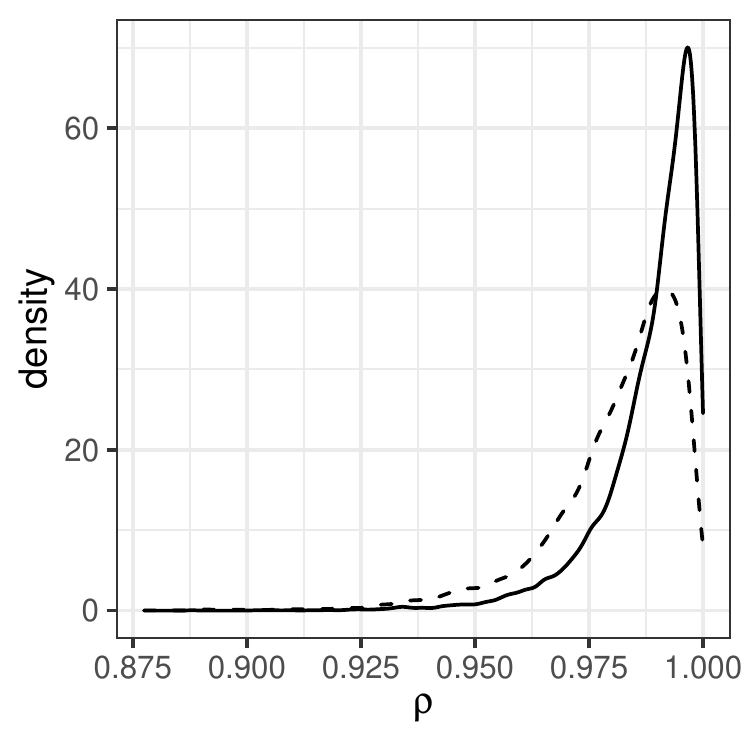}
\caption{Empirical density from direct sampler (solid line) and Metropolis-Hastings (dashed line).}
\label{fig:car-rho-mcmc-density}
\end{subfigure}
\caption{Draws of $\rho$ from Gibbs sampler with direct sampling step versus Metropolis-Hastings step.}
\label{fig:car-rho-mcmc}
\end{figure*}

\subsection{Sampling the Degrees-of-Freedom in Robust Regression}
\label{sec:example-tdist}

The Student t-distribution may be considered as an alternative to the Normal distribution when additional variability is needed in a linear model. Let $t_\nu$ denote a t-distribution with degrees of freedom $\nu$ and density function
\begin{align*}
f(x \mid \nu) =
\frac{\Gamma((\nu + 1) / 2)}{\Gamma(\nu / 2) \sqrt{\nu \pi}} \left(
1 + \frac{x^2}{\nu} \right)^{-(\nu+1)/2}.
\end{align*}
Suppose outcomes $y_i = \vec{x}_i^\top \vec{\beta} + \sigma \epsilon_i$ are observed for $i = 1, \ldots, n$, where $\epsilon_1, \ldots, \epsilon_n \iid t_\nu$ and $\vec{x}_i \in \mathbb{R}^d$ are given covariates. Using a particular data augmentation and taking $\nu$ to be fixed, it is possible to formulate a Gibbs sampler whose steps consist of drawing from standard distributions \citep{GelmanEtAl2013}. With $\nu$ not fixed and inference desired on $\btheta = (\vec{\beta}, \sigma, \nu)$, \citet{Geweke1994} proposes a rejection sampler for the conditional of $\nu$. Let us illustrate how a direct sampler can also be used to effectively generate draws from this nonstandard conditional distribution.

An augmented version of the model assumes variables $s_1, \ldots, s_n$ with $\vec{D}_s = \Diag(s_1, \ldots, s_n)$ such that
\begin{align}
\vec{y} = \vec{X} \vec{\beta} + \vec{\gamma}, \quad
\vec{\gamma} \sim \text{N}(\vec{0}, \vec{D}_s), \quad
s_i \iid \text{IG}(\nu / 2, \nu \sigma^2 / 2),
\label{eqn:t-regression}
\end{align}
for $i = 1, \ldots n$ with prior distributions $\vec{\beta} \sim \text{N}(\vec{0}, \sigma_\beta^2 \vec{I})$ and $\sigma^2 \sim \text{Gamma}(a_\sigma, b_\sigma)$. Furthermore, let $\nu$ have a $\text{Uniform}(a_\nu, b_\nu)$ prior. The joint distribution of all random variables is
\begin{align*}
&f(\vec{y}, \vec{s}, \mu, \sigma^2, \nu) = \\
&\quad (2 \pi)^{-n/2}
\left[ \prod_{i=1}^n s_i^{-1/2} \right]
\exp\left\{ -\frac{1}{2} \sum_{i=1}^n \frac{(y_i - \vec{x}_i^\top \vec{\beta})^2}{s_i} \right\} \\
&\quad \times
\frac{(\nu \sigma^2 / 2)^{n \nu / 2}}{\Gamma(\nu / 2)^n}
\left[ \prod_{i=1}^n s_i \right]^{-\frac{\nu}{2} - 1}
\exp\left\{-\frac{\nu \sigma^2}{2} \sum_{i=1}^n \frac{1}{s_i} \right\} \\
&\quad \times
\left[ 2\pi\sigma_\beta^2 \right]^{-d/2}
\exp\left\{ -\frac{1}{2 \sigma_\beta^2} \vec{\beta}^\top \vec{\beta} \right\} \\
&\quad \times
\frac{b_\sigma^{a_\sigma}}{\Gamma(a_\sigma)} (\sigma^2)^{a_\sigma - 1}
e^{-b_{\sigma} \sigma^2}
\ind(a_\nu \leq \nu \leq b_\nu).
\end{align*}
This distribution yields conditionals
\begin{align*}
&[\vec{\beta} \mid \vec{s}, \sigma^2, \nu, \vec{y}] \sim \text{N}\left(
\vec{\vartheta}_\beta, \vec{\Omega}_\beta^{-1}
\right), \\
&[\sigma^2 \mid \vec{s}, \vec{\beta}, \nu, \vec{y}] \sim \text{Gamma}\left(
a_\sigma + \frac{n \nu}{2}, b_\sigma + \frac{\nu}{2} \sum_{i=1}^n \frac{1}{s_i}
\right), \\
&[\vec{s} \mid \sigma^2, \vec{\beta}, \nu, \vec{y}] = \prod_{i=1}^n \text{IG}\left(s_i \given \frac{\nu + 1}{2}, \frac{\nu \sigma^2}{2} + \frac{(y_i - \vec{x}_i^\top \vec{\beta})^2}{2} \right),
\end{align*}
having familiar forms so that draws are straightforward, where
\begin{align*}
\vec{\Omega}_\beta = \vec{X}^\top \vec{D}_s^{-1} \vec{X} + \sigma_\beta^{-2} \vec{I},
\quad
\vec{\vartheta}_\beta = \vec{\Omega}_\beta^{-1} \vec{X}^\top \vec{D}_s^{-1} \vec{y},
\end{align*}
and $\vec{X} \in \mathbb{R}^{n \times d}$ is the matrix with rows $\vec{x}_1, \ldots, \vec{x}_n$. More interesting is the distribution of $[\nu \mid \vec{s}, \vec{\beta}, \sigma^2, \vec{y}]$, which has the form
\begin{align}
f(\nu \mid \vec{s}, \vec{\beta}, \sigma^2, \vec{y}) \propto w(\nu) g(\nu)
\label{eqn:nu-conditional}
\end{align}
with $\text{Uniform}(a_\nu, b_\nu)$ base distribution $g(\nu) = (b_\nu - a_\nu)^{-1} \cdot \ind(a_\nu \leq \nu \leq b_\nu)$ and weight function $w$ such that
\begin{align*}
&\log w(\nu) \\ 
&\quad= \frac{n \nu}{2} \log (\nu / 2) - n \log \Gamma(\nu / 2)
- \frac{\nu}{2} \sum_{i=1}^n \log \frac{s_i}{\sigma^2} \\
&\quad\quad- \frac{\nu}{2} \sum_{i=1}^n \frac{\sigma^2}{s_i}
+ \log \ind(a_\nu \leq \nu \leq b_\nu) \\
&\quad= n \left[ \frac{\nu}{2} \log \frac{\nu}{2} - \log \Gamma\left(\frac{\nu}{2}\right) \right] -
A \nu + \log \ind(a_\nu \leq \nu \leq b_\nu)
\end{align*}
where $A = \frac{1}{2} \sum_{i=1}^n \log (s_i / \sigma^2) + \frac{1}{2} \sum_{i=1}^n \sigma^2 / s_i$. Temporarily disregarding the indicator $\ind(a_\nu \leq \nu \leq b_\nu)$ and considering $\nu \in (0, \infty)$,
\begin{align}
\frac{d}{d \nu} \log w(\nu) &= 
\frac{n}{2} \left[ \log\left( \frac{\nu}{2} \right) - \psi\left( \frac{\nu}{2} \right) \right] + n/2 - A.
\label{eqn:tdist-weight-deriv}
\end{align}
It can be shown \citep[e.g.][]{Alzer1997} that
\begin{align*}
&\frac{1}{2x} < \log x - \psi(x) < \frac{1}{x}, \quad x > 0, \\
&\lim_{x \rightarrow 0} \frac{\log x - \psi(x)}{1/x} = 1, \quad
\lim_{x \rightarrow \infty} \frac{\log x - \psi(x)}{1/x} = \frac{1}{2}.
\end{align*}
Therefore $\log\left( \frac{\nu}{2} \right) - \psi(\frac{\nu}{2})$ is positive, decreases to 0 as $\nu$ increases, and increases as $\nu$ decreases to 0. Furthermore, the function $g(x) = \frac{1}{2} \log (x / \sigma^2) + \frac{1}{2} (\sigma^2 / x)$ is minimized by $x = \sigma^2$ so that $A = \sum_{i=1}^n g(s_i) \geq n/2$. Notice that \eqref{eqn:tdist-weight-deriv} has a root in $(0, \infty)$ when $A > n/2$, and has no root if $A = n/2$. We can gather some information about the behavior of $\log w(\nu)$ from \eqref{eqn:tdist-weight-deriv}.
\begin{enumerate}
\item When $\frac{d}{d \nu} \log w(\nu)$ has no root, it is always positive so that $\log w(\nu)$ is an increasing function. Here, $\log c = \log w(b_\nu)$.

\item When $\frac{d}{d \nu} \log w(\nu)$ has a root, $\log w(\nu)$ has a single maximizer $\nu^*$. Therefore, $\log w(\nu)$ is unimodal on $[a_\nu, b_\nu]$ with $\log c = \log w(\nu^*)$ if $\nu^* \in [a_\nu, b_\nu]$. If $\nu^* > b_\nu$, $\log w(\nu)$ is an increasing function on $[a_\nu, b_\nu]$ with $\log c = \log w(b_\nu)$. Otherwise, $\nu^* < a_\nu$, and $\log w(\nu)$ is a decreasing function on $[a_\nu, b_\nu]$ with $\log c = \log w(a_\nu)$.
\end{enumerate}
Numerical root finding such as Algorithm~\ref{alg:bisection} may be used to compute the endpoints $\{ \nu_1(u), \nu_2(u) \}$ of the interval $A_u(\nu) = \{ \nu > 0 : \log w(x) > \log (uc) \}$. If there is a maximizer $\nu^*$ in the interval $[a_\nu, b_\nu]$, $\nu_1$ will be found in $[a_\nu, \nu^*]$ and $\nu_2$ will be found in $[\nu^*, b_\nu]$. If $\log w(\nu)$ is strictly increasing, $\nu_2 = b_\nu$ and $\nu_1$ is found in $[a_\nu, b_\nu]$. Otherwise, if $\log w(\nu)$ is strictly decreasing, $\nu_1 = a_\nu$ and $\nu_2$ is found in $[a_\nu, b_\nu]$.

Notice that a bounded prior for $\nu$ is needed to obtain a finite maximum value $c$ of the weight function; therefore, our choice of Uniform prior is a departure from the Exponential prior assumed by \citet{Geweke1994}. A rejection sampler similar to Geweke's can be obtained by following the original derivation with several minor differences. First, Geweke's constant $A$ features an additional term with the Exponential hyperparameter which is now absent. Second, we take the proposal distribution to be a truncated Exponential distribution with density $q(x \mid \alpha, a_\nu, b_\nu) \propto \alpha e^{-x \alpha} \ind(a_\nu \leq x \leq b_\nu)$ rather than an untruncated Exponential distribution $q(x \mid \alpha) \propto \alpha e^{-x \alpha}$. Constraining $\alpha = 1/\nu$, let $\nu^*$ be the value of $\nu$ which maximizes the ratio $f(\nu \mid \vec{s}, \vec{\beta}, \sigma^2, \vec{y}) / q(\nu \mid \alpha, a_\nu, b_\nu)$; this $\nu$ satisfies
\begin{align}
\frac{n}{2} \left[ \log\left(\frac{\nu}{2}\right) + 1 - \psi\left(\frac{\nu}{2}\right) \right] + \frac{1}{\nu} - A = 0.
\label{eqn:geweke-root}
\end{align}
Algorithm~\ref{alg:geweke} gives a rejection sampler based on $q$ to generate candidates and the maximized ratio to determine when to accept.

Figure~\ref{fig:tdist_draws} compares the empirical density of 100,000 draws from the direct sampler with $N$ knots initially selected using Algorithm~\ref{alg:geweke}. The values $n = 200$, $a_\nu = 0.01$, and $b_\nu = 200$ are fixed and $A$ is varied to take on values 101, 120, 200, and 400. As expected, both samplers generate draws from the same target distribution. Table~\ref{tab:tdist-rejections} shows the number of rejections to obtain 100,000 draws for both samplers, now including $N \in \{ 5, 20, 50, 100 \}$ initially selected knots. Here it is apparent that Algorithm~\ref{alg:geweke} rejects on the order of ten candidates for each saved variate while the direct sampler rejects for less than 1\% of draws on average. However, within a practical Gibbs sampling setting, Algorithm~\ref{alg:geweke} may still be faster because each step requires very little computation.

\begin{algorithm}
\caption{Rejection sampling based on \citet{Geweke1994}.}
\label{alg:geweke}
\begin{enumerate}
\item Let $\nu^*$ be the value of $\nu$ which satisfies \eqref{eqn:geweke-root}.

\item
\label{step:candidate}
Draw candidate $\nu$ from the truncated Exponential distribution $q(x \mid 1/\nu^*, a_\nu, b_\nu)$.

\item Draw $\omega \sim \text{Uniform}(0,1)$ and accept $\nu$ as a draw from \eqref{eqn:nu-conditional} if
\begin{align*}
\omega <
\frac{
\left( \frac{\nu}{2} \right)^{n \nu / 2}
\left[ \Gamma\left(\frac{\nu}{2}\right) \right]^{-n}
\exp(-\nu A + \nu / \nu^*)
}{
\left( \frac{\nu^*}{2} \right)^{n \nu^* / 2}
\left[ \Gamma\left(\frac{\nu^*}{2}\right) \right]^{-n}
\exp(-\nu^* A + 1)
};
\end{align*}
otherwise, reject $\nu$ and go to step \ref{step:candidate}.
\end{enumerate}
\end{algorithm}

\begin{figure*}
\centering
\begin{subfigure}{0.24\textwidth}
\includegraphics[width=\textwidth]{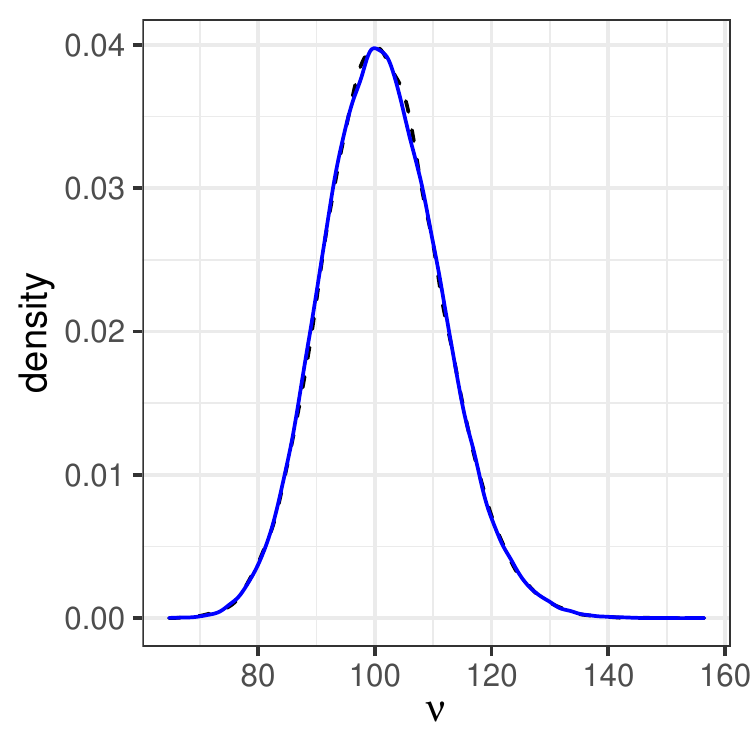}
\caption{$A = 101$.}
\label{fig:tdist-draws-A1}
\end{subfigure}
\begin{subfigure}{0.24\textwidth}
\includegraphics[width=\textwidth]{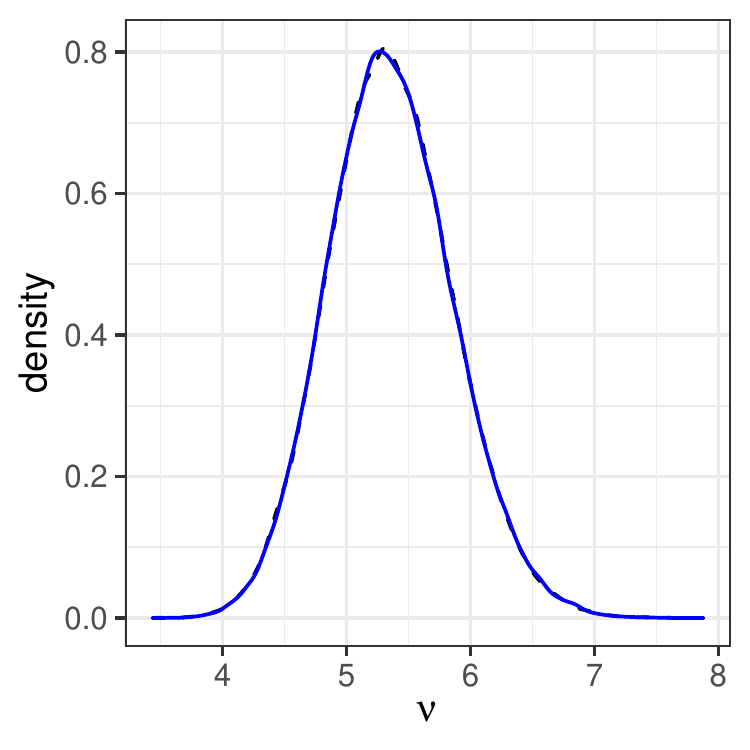}
\caption{$A = 120$.}
\label{fig:tdist-draws-A2}
\end{subfigure}
\begin{subfigure}{0.24\textwidth}
\includegraphics[width=\textwidth]{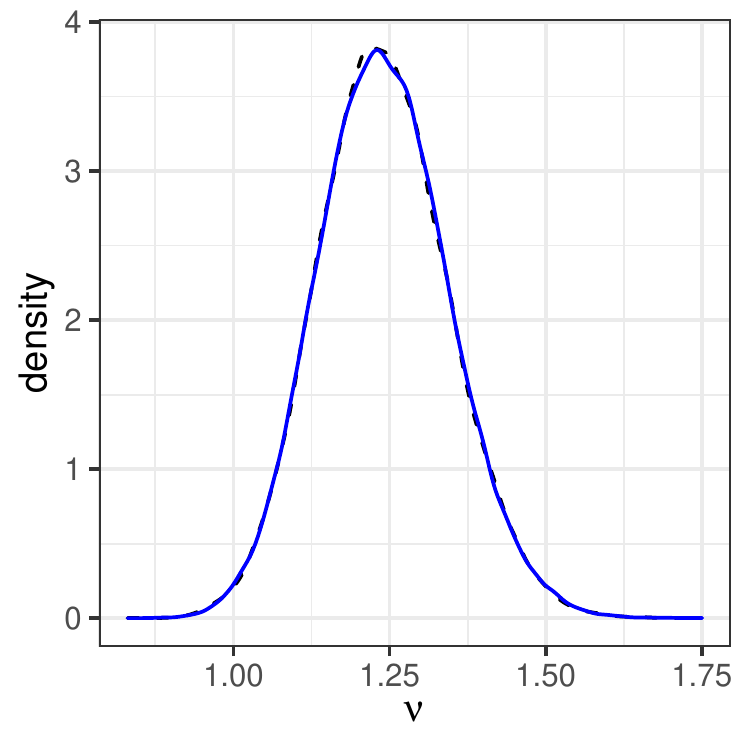}
\caption{$A = 200$.}
\label{fig:tdist-draws-A3}
\end{subfigure}
\begin{subfigure}{0.24\textwidth}
\includegraphics[width=\textwidth]{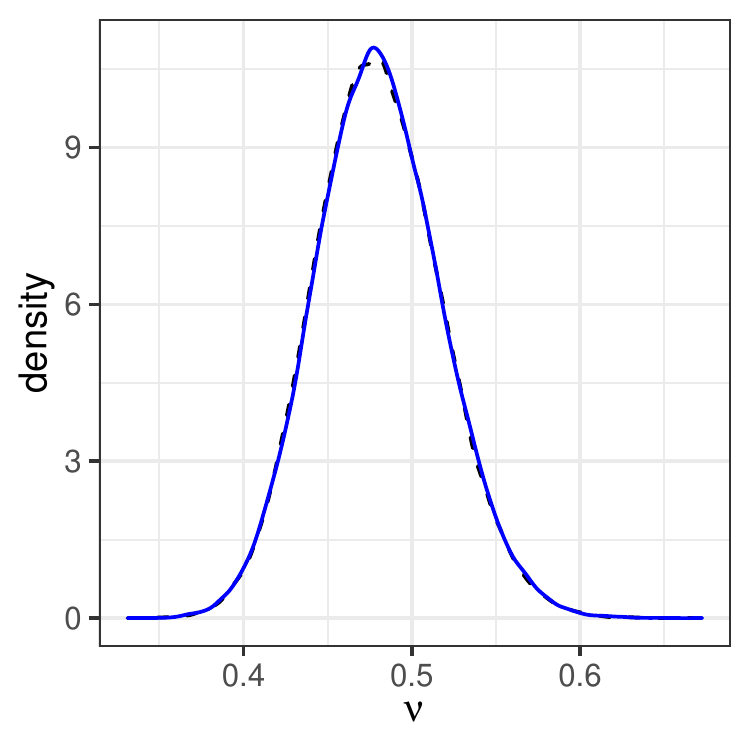}
\caption{$A = 400$.}
\label{fig:tdist-draws-A4}
\end{subfigure}
\caption{Empirical density of 100,000 draws from \eqref{eqn:nu-conditional} with sample size $n = 200$ and $A$ specified in the subcaption. Dashed black curve represents Algorithm~\ref{alg:geweke} and solid blue curve represents direct sampler with $N = 5$ initial knots.}
\label{fig:tdist_draws}
\end{figure*}

\begin{table}
\centering
\caption{Rejected candidates to obtain 100,000 draws using Algorithm~\ref{alg:geweke} and the direct sampler with $N$ initial knots.}
\label{tab:tdist-rejections}
\begin{tabular}{rrrrrr}
\toprule
\multicolumn{1}{c}{} &
\multicolumn{4}{c}{Direct Sampler} &
\multicolumn{1}{c}{} \\
\multicolumn{1}{c}{$A$} &
\multicolumn{1}{c}{$N = 5$} &
\multicolumn{1}{c}{$N = 20$} &
\multicolumn{1}{c}{$N = 50$} &
\multicolumn{1}{c}{$N = 100$} &
\multicolumn{1}{c}{Algorithm~\ref{alg:geweke}} \\
\cmidrule(lr){1-1}
\cmidrule(lr){2-5}
\cmidrule(lr){6-6}
101 & 608 & 647 & 589 & 495 &   841,390 \\
120 & 643 & 605 & 581 & 496 & 1,049,358 \\
200 & 622 & 575 & 549 & 523 & 1,173,088 \\
400 & 614 & 564 & 581 & 533 & 1,273,444 \\
\bottomrule
\end{tabular}
\end{table}

\citet{LangeEtAl1989} provide a number of interesting examples of regression analyses using t-distributed errors. In particular, their Example 3 studies the relationship between two measurements of blood flow in the canine myocardium. The variable $r_i$ measures regional myocardial blood flow from an invasive procedure, while $y_i$ is a measurement obtained using positron emission tomography within $n$ cases indexed $i = 1, \ldots, n$. It is assumed that
\begin{align*}
&y_i = \mu(r_i) + \sigma \epsilon_i, \quad
\epsilon_i \sim t_{\nu}(0,1), \\
&\mu(r) = r \{1 - \phi_1 \exp( -\phi_2 / r) \},
\end{align*}
for parameters $\vec{\phi} = (\phi_1, \phi_2)$. We consider a simulated dataset based on this setting, with $n = 200$ and $r_i \iid \text{Uniform}(0,10)$. Data generating values of parameters $\vec{\phi}$ are to taken to be $\phi_1 = 0.746$ and $\phi_2 = 274.7$, based on estimates reported in \citet{LangeEtAl1989}, while $\nu = 2$ and $\sigma = 1.25$ are taken to be the degrees of freedom and scale for the random errors. To fit linear model \eqref{eqn:t-regression}, the $i$th row $\vec{x}(r_i)$ of design matrix $\vec{X}$ is obtained from a cubic polynomial basis using the bs function in the R splines package \citep{Rcore2022}. We apply the Gibbs sampler with direct sampling to draw $\nu$ using $N = 30$ initial knots. A chain of 10,000 iterations is computed with first 5,000 discarded as as burn-in sample. Hyperparameters are taken to be $a_\nu = 0.01$, $b_\nu = 200$, $a_\sigma = 1$, $b_\sigma = 1$, and $\sigma_\beta^2 = 100$. Table~\ref{tab:tdist-fit} summarizes the saved draws of $\btheta$, while Figure~\ref{fig:tdist-fit} compares the fitted function to the true function $\mu(r)$. The model appears to be capturing the data-generating values of $\sigma$, $\nu$, and $\mu(r)$ appropriately. We note that total sampling time was 12.2 seconds, of which 10.3 seconds was spent drawing $\nu$ with the direct sampler. There were 260 rejections in the 10,000 draws of $\nu$.

\citet{Hosszejni2021-arxiv} presents a recent survey for Bayesian inference of $\nu$. Here it is noted that the approach of \citet{Geweke1994}---considered in the present section---works well for small $\nu$ but mixing tends to worsen for larger $\nu$. Therefore, other sampling strategies are recommended when $\nu$ may be larger.

\begin{table}
\centering
\caption{Summary of posterior distribution.}
\label{tab:tdist-fit}
\begin{tabular}{crrrrr}
\toprule
\multicolumn{1}{c}{} &
\multicolumn{1}{c}{Mean} &
\multicolumn{1}{c}{SD} &
\multicolumn{1}{c}{$2.5\%$} &
\multicolumn{1}{c}{$97.5\%$} \\
\midrule
$\beta_1$  & -0.7677 & 0.5734 & -1.8974 &  0.3453 \\
$\beta_2$  &  3.9061 & 1.4799 &  1.0359 &  6.7897 \\
$\beta_3$  &  8.4826 & 0.7902 &  6.8905 & 10.0353 \\
$\beta_4$  & 10.4723 & 0.8227 &  8.8485 & 12.0764 \\
$\sigma^2$ &  1.4777 & 0.3002 &  0.9868 &  2.1660 \\
$\nu$      &  2.3409 & 0.5045 &  1.5891 &  3.5602 \\
\bottomrule
\end{tabular}
\end{table}

\begin{figure}
\centering
\includegraphics[width=0.40\textwidth]{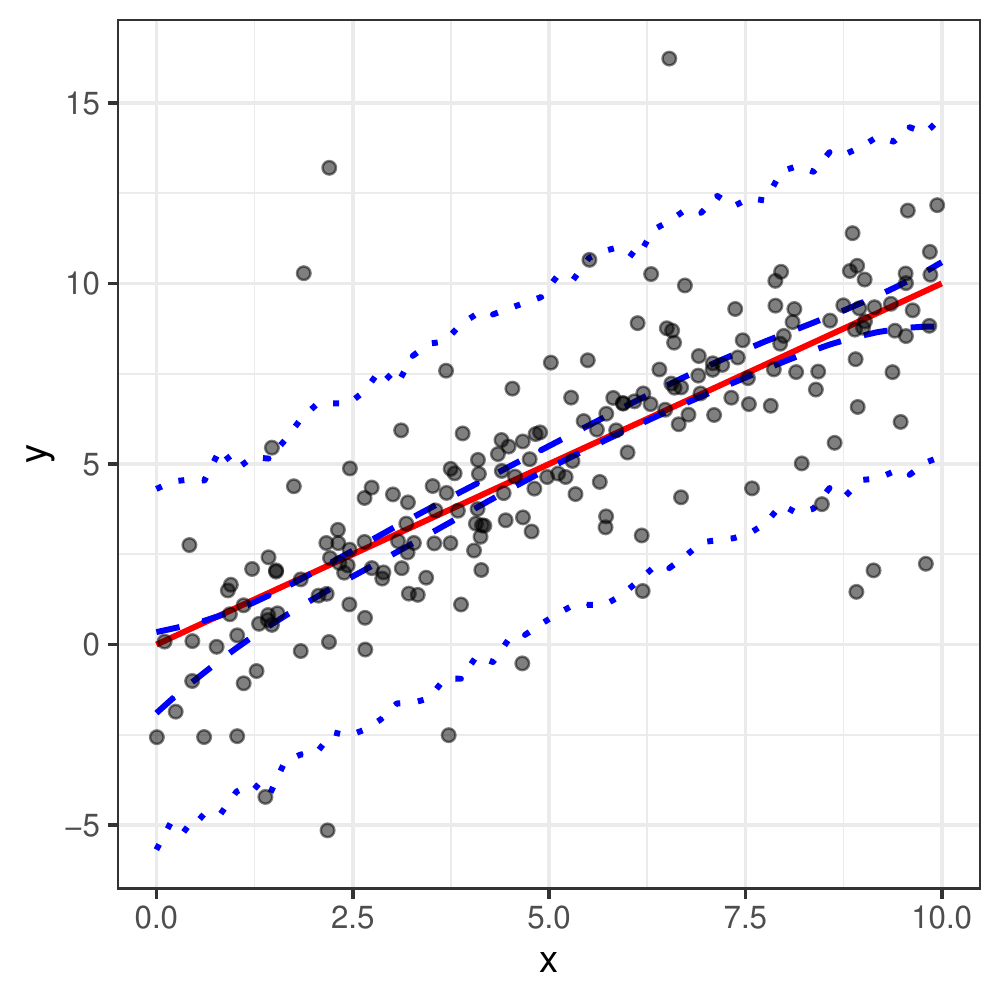}
\caption{The true function $\mu(r)$ (solid red curve), pointwise 95\% credible interval for $\mu(r)$ (dashed blue curve), and pointwise 95\% interval for posterior predictive distribution of $y(r) = \mu(r) + \epsilon$, $\epsilon \sim t_\nu(0, \sigma)$ (dotted blue line), for $r \in [0, 10]$. Observed $y_1, \ldots, y_n$ are shown as block dots.}
\label{fig:tdist-fit}
\end{figure}

\section{Discussion and Conclusions}
\label{sec:discussion}

The density $p(u)$ which arises in direct sampling \citep{WalkerEtAl2011} is monotone, nonincreasing on $[0,1]$, and subject to sudden jumps. This motivated us to consider step functions to approximate $p(u)$. Useful samplers may be obtained for some univariate target distributions where $A_u = \{x \in \Omega : w(x) > cu \}$ is an interval, and may further be combined with rejection sampling to generate exact draws with a small number of rejections. Examples in Sections~\ref{sec:example-car} and \ref{sec:example-tdist} illustrated sampling from non-standard conditionals within in a Gibbs sampler. All three examples already have practical samplers described in the literature; our proposed sampler may be useful when encountering unfamiliar weighted distributions where no such method is readily available.

Care is required in the implementation of the sampler; e.g., the possibility of encountering very small magnitude floating point numbers motivates use of the geometric midpoint and carrying out many of the calculations on the log-scale. The idea may be extended to settings where $A_u$ is a more complicated set such as a union of intervals, provided that the endpoints can be identified without too much computation. Multivariate settings may also be possible, provided that $A_u$ is not too difficult to characterize and draws from $f(\vec{x} \mid u) \propto g(\vec{x}) \ind(\vec{x} \in A_u)$ can be reliably generated.

\section*{Acknowledgements}
The author is grateful to Drs.~Scott Holan, Kyle Irimata, Ryan Janicki, and James Livsey at the U.S. Census Bureau for discussions which motivated this work.


\end{document}